\documentclass[natbib,smallcondensed]{svjour3}
\usepackage{tikz}
\usepackage{amssymb}
\usepackage{amsfonts}
\usepackage{amsmath}
\usepackage{natbib}
\usepackage[vlined,ruled]{algorithm2e}
\usepackage{graphicx}

\def\F{\ensuremath{\mathcal{F}}} 
\def\N{\ensuremath{\mathbb{N}}}  
\def\K{\ensuremath{\mathbb{F}}}  
\def\C{\ensuremath{\mathcal{C}}} 

\newtheorem{notation}{Notation}

\begin{document}
\sloppypar
\title{Properties and constructions of coincident functions}

\author{Morgan Barbier \and
  Hayat Cheballah \and
  Jean-Marie Le Bars}

\institute{M. Barbier
 \and Hayat Cheballah
 \and J.-M. Le Bars
 \at ENSICAEN - UNICAEN - GREYC \\
  \email{morgan.barbier@ensicaen.fr}\\
  \email{hayat.cheballah@gmail.com}\\
  \email{jean-marie.lebars@unicaen.fr}
}

\maketitle
\begin{abstract} 
Extensive studies of Boolean functions are carried in many fields. The
Mobius transform is often involved for these studies. In particular,
it plays a central role in coincident functions, the class of Boolean
functions invariant by this transformation.  This class -- which has
been recently introduced -- has interesting properties, in particular
if we want to control both the Hamming weight and the degree. We
propose an innovative way to handle the Mobius transform which allows
the composition between several Boolean functions and the use of
Shannon or Reed-Muller decompositions.  Thus we benefit from a better
knowledge of coincident functions and introduce new properties.  We
show experimentally that for many features, coincident functions look
like any Boolean functions. 
\end{abstract}

\section{Introduction}
\label{sec:intro}
Numerous studies with Boolean functions have been conducted in various
fields like cryptography and error correcting codes
(\cite{carletBook}), Boolean circuits and Boolean Decision Diagram
(\cite{bryant86}), Boolean logic (\cite{boole1848}) or constraint
satisfaction problems (\cite{creignou2001}).  There are many ways to
represent a Boolean function which depends of the domain. For instance,
on propositional logic one usually uses the conjunctive normal form or
the disjunctive normal form,  while we often use the BDD in Boolean
circuits.  Most of the time these studies involve several
criteria. In cryptography, the (algebraic) degree and the (Hamming)
weight are crucial criteria.\\

Unfortunately, the best representation for the degree is the Algebraic
Normal Form (sum of monomials), while the weight requires the truth table (sum of
minterms). Thus the Reed-Muller decomposition (or expansion) allows
us to perform recursive decomposition (\cite{Kasami70}), enumeration
and random generation among the degree whereas the Shannon
decomposition (or expansion) does the same task among the weight
\cite{shannon1949} shows the switching network
interpretation of this identity, but (\cite{boole1854}) will be the first to
mentioned it.  These decompositions allow us to
decompose a Boolean function with $n$ variables into two Boolean
functions with $n-1$ variables or equivalently to build a Boolean function with
$n$ variables with two Boolean functions with $n-1$ variables.\\

Since these decompositions appear
to be orthogonal, it seems unreachable to consider them simultaneously
or to perform enumeration or random generation with both
criteria. This is why we propose the study of coincident functions
for which we have a correspondence between the monomials and the
minterms.\\

In \cite{coincidentFunc}, the author defined for the first time the
class of coincident functions, the boolean functions invariant by
Mobius transform (see \cite{guillot} phd Thesis for a deep study of this
property). We revisit their results with a new point of view more
convenient, in particular we introduce a Mobius transform conditioned
by the number of variables which gives a recursive (Reed-Muller and
Shannon) decomposition of Mobius transform, linking the two
decompositions. We investigate the distribution of degree and weight
and provide a uniform random generator of these functions. We also
benefit to the structure of lattice of the valuations to study the
monotonic coincident functions and to provide a complete construction
of symmetric coincident ones.\\

The paper is organized as follows. We recall the basic definitions
related to boolean functions and to Mobius transform in
Section~\ref{sec:preliminaries}. We introduce our new point of view
and the resulting properties concerning Mobius transform and
coincident in Section~\ref{sec:properties}. Thanks to lattices, we
exhibit, in Section~\ref{sec:lattices}, links between coincident
functions with monotonic and symmetric functions, and also the random
generation of coincident functions. Finally, we propose in
Section~\ref{sec:experiment} a set of experiment resulting on
coincident functions.

\section{Definition and first properties of Boolean functions and Mobius transform}
\label{sec:preliminaries}

\subsection{Boolean function}

Let $\F_n$ be the set of the boolean functions with $n$
variables. Monomials and minterms play a role of canonical element in
the different writings.
\begin{definition}[Monomials and minterms]
Let us to denote $x=(x_1,\dots,x_n)$. For any $u = (u_1,\ldots, u_n) \in \K_2^n,$ 
$x^u$ will be denoted the monomial $x_1^{u_1}\dots x_n^{u_n}$. The
minterm $M_u$ is the Boolean function with $n$ variables defined by
\[
M_u(a) = 
\left\lbrace
\begin{array}{ll}
  1, & \mbox{ if } u = a;\\
  0, & \mbox{ otherwise.}
\end{array}
\right.
\]
\end{definition}

The two following definitions provide two different point of views what we wish
to link.
\begin{definition}[Algebraic Normal Form -- ANF]
A boolean function $f$ can be viewed as the exclusive sum of a subset of the set of monomials 
in variables $x_1, \ldots, x_n$.
\[ 
f = \bigoplus_{u \in \F_2^n} \alpha_u x^u.
\]
\end{definition}

\begin{definition}[Truth table or valuations of $f$]
A boolean function $f$ can be viewed as the exclusive sum of a subset of the set of minterms in variables 
$x_1, \ldots, x_n$.
\[
f = \bigoplus_{u \in \F_2^n} \beta_u M_u.
\]
\end{definition}
\begin{notation}
Let $u$ and $v \in \K_2^n$. We will write $u \preceq v$ when 
$u_i \le v_i$, for any \mbox{$i \in \{1,\ldots, n\}$} and 
$u \prec v$ when $u \preceq v$ and $u \not= v$. 
\end{notation}

The previous partially ordered set provides connection between these
notions.
\begin{proposition}
 Let $u\in \K_2^n$, then
\begin{equation}\label{monomial2minterm}
x^u = \bigoplus_{u \preceq v} M_v.
\end{equation}
\end{proposition}
\begin{proof}
Let $f = x^u$ and $a \in \K_2^n$, $f(a) = 1$ if and only if $a_i = 1$
for all $i$ such that $u_i = 1$, {\it ie} $u \preceq a$. Hence $f
= \bigoplus_{u \preceq v} M_v$.
\end{proof}
 \begin{proposition}
 Let $u\in \K_2^n$, then
\begin{equation}\label{minterm2monomial}
M_u = \bigoplus_{u \preceq v} x^v.
\end{equation}
\end{proposition}
\begin{proof}
Let $u = (u_1, \ldots,u_n) \in \K_2^{n}$.
It is easily seen that $M_u$ is equal to the product $\prod_{i = 1}^n (x_i \oplus u_i \oplus 1)$. 
Let $I_0(u) = \{ i \in \{1,\ldots, n\} : u_i = 0\}$ and \mbox{$I_1(u)
= \{ i \in \{ 1, \ldots, n\} : u_i = 1\}$}, it follows
\[
\begin{array}{lcl}
M_u & = & \prod_{i \in I_1(u)} x_i  \prod_{i \in I_0(u)} (1 \oplus x_i)\\
  & = & \bigoplus_{u \preceq v} x^v.
\end{array}
\]
\end{proof}

 \begin{notation}
There are several natural ways to encode $f$ by a binary word of length $2^n$. 

We will choose the following natural encodings directly derived by the
previous definition.
\[
T(f) = t_1 \ldots t_{2^n},
\]
where $t_k$, with  $\sum_{i=1}^n u_i \; 2^{i-1}$ the $2$-adic representation of $k$
and $t_k = \beta_u$.
\[
A(f) = a_1 \ldots a_{2^n},
\]
where $a_k = \alpha_{u}$.

We denote by $\psi$ the bijection $\psi (A(f)) = T(f)$.
\end{notation}

The two definitions below introduce crucial parameters of a boolean function.
\begin{definition}[Hamming weight]
  Let $f\in \F_n$ be a boolean function, we will write $w_H(f)$ the (Hamming) weight of $f$, 
{\it ie} the number of 1 of $T(f)$.
\end{definition}
\begin{definition}[Algebraic degree]
 Let $f\in \F_n$ be a boolean function, we will write $d(f)$ the (algebraic) degree of $f$, 
{\it ie} the maximal degree of the monomials in the ANF of $f$.
\end{definition}

While the Reed-Muller decomposition is related to the algebraic normal
form, the Shannon one is associated to truth table.
\begin{definition}[Reed-Muller decomposition]
  Let $f \in \F_n$. The {\em Reed-Muller decomposition}, consists in
  rewriting the boolean function as
  $$
  f = f_R^0 \oplus x_n f_R^1,
  $$
  where $f_R^0, f_R^1 \in \F_{n-1}$ and are unique.
\end{definition}

\begin{definition}[Shannon decomposition]
  Let $f \in \F_n$. The {\em Shannon decomposition}, consists in
  rewriting the boolean function as
  $$
  f = (1+x_n)f_S^0 \oplus x_n f_S^1,
  $$
  where $f_S^0, f_S^1 \in \F_{n-1}$ and are unique.
\end{definition}

\begin{remark}
    Let $f \in \F_n$. Clearly, we have by identification $f_R^0=f_S^0$
    and $f_R^1 = f_S^0 \oplus f_S^1$.
\end{remark}

\begin{remark}
The Shannon decomposition is the natural decomposition for manipulating 
the minterms since $T(f) = T(f_S^0)\ |\ T(f_S^1)$, where $|$ denotes the concatenation.
This trivially implies 
\[
w_H(f) = w_H(f_S^0) + w_H(f_S^1).
\]
\end{remark}

\begin{remark}
The Reed-Muller decomposition is the natural decomposition for manipulating 
the monomials since $A(f) = A(f_R^0)\ |\ A(f_R^1)$.
This implies 
\[
d(f) = max (d(f_R^0), d(f_R^1)+1).
\]
\end{remark}

\begin{notation}
From now one, we will write $0_n$ the valuation of $\K_2^n$ $(0, \ldots, 0)$ 
and $\mathbf 0_n$ (resp. $\mathbf 1_n$) the Boolean function with $n$ variables which takes the value $0$ (resp. $1$) 
for any valuation. 
\end{notation}

\subsection{Mobius transform}
\label{ssec:mobius}
The Mobius transform is at the heart of this article.
\begin{definition}[Mobius transform]
  The {\em Mobius transform}, noted $\mu$ is defined by the following
  bijection
  $$
  \begin{array}{rcl}
    \mu : \F_n & \longleftrightarrow & \F_n\\
    f & \longmapsto & \mu(f),\\
  \end{array}
  $$
  such that, for any $f \in \F_n$
  $$
    f = \bigoplus_{u \in \K_2^n} \mu(f)(u) x^u.
  $$
\end{definition}
The statements of the following proposition derives directly from the definition of the Mobius transform.
\begin{proposition} \label{prop:firstMobius}
\cite[Theorem 1 and Lemma 3]{coincidentFunc}
Let $f, g \in \F_n$ and $u \in \K_2^n$, 
$$
\mu(f\oplus g) = \mu(f) \oplus \mu(g) \mbox{ and }  \mu(f)(0_n) = f(O_n).
$$ 
\end{proposition}
\begin{proof}
  Both results are direct implication of the definition of $\mu$,  
  let \mbox{$f = \bigoplus_{u \in \K_2^n} \mu(f)(u) x^u$} and  $g = \bigoplus_{u \in \K_2^n} \mu(g)(u) x^u$. 
  Then \mbox{$f \oplus g = \bigoplus_{u \in \K_2^n} (\mu(f)(u)\oplus \mu(g)(u)) x^u$} and 
  $\mu(f \oplus g)(u) = \mu(f)(u) \oplus \mu(g)(u)$.
  Therefore let $a = 0_n$, $a^u = 0$ for any $a \prec u$, then $f(0_n) = \mu(f)(a)$.
\end{proof}
\begin{proposition} \label{prop:muminterm}
\cite[Lemma 2]{coincidentFunc}
Let $u \in \K_2^n$, then $\mu(x^u) = M_u$ and $\mu(M_u) = x^u$.
\end{proposition}
\begin{proof}
Let $u = (u_1, \ldots,u_n) \in \K_2^{n}$ and $f= x^u$. 
$f = \bigoplus_{v \in \K_2^n} \mu(f)(v) x^v$. Then $\mu(f) (v)$ is equal to $1$ if $v = u$ and equal to 
$0$ otherwise, which is the definition of $M_u$. 
\[
\begin{array}{lcl}
M_u & = & \bigoplus_{u \preceq v} x^v \mbox{ by } (\ref{minterm2monomial})\\
\mu(M_u) & = & \bigoplus_{u \preceq v} \mu(x^v)\\
& = & \bigoplus_{u \preceq v} M_v\\
& = & x^u \mbox{ by } (\ref{monomial2minterm}).\\
\end{array}
\]
\end{proof}

\begin{proposition}\label{involution}
The Mobius transform is a involution, {\it ie} $\mu^2(f) = f$ and $\mu(f) = g$ if and only if $\mu(g) = f$.
\end{proposition}
\begin{proof}
\[
\begin{array}{lcl}
f & = & \bigoplus_{u \in \K_2^n} \mu(f)(u) x^u\\
\mu(f) & = & \bigoplus_{u \in \K_2^n} \mu(f)(u) M_u\\
\mu^2(f) & = & \bigoplus_{u \in \K_2^n} \mu(f)(u) x^u.
\end{array}
\]
\end{proof}

\section{New properties on Mobius transform and coincident functions}
\label{sec:properties}

Let us start this section by the following remark, which is one of
main statement of this paper.
\begin{notation}
Let $f \in \F_n$, we will denote by $P(f)$ the polynomial form of $f$.
\end{notation}
\begin{remark}
 \label{rem:FnCFnk}
The polynomial form contains only the variables which play a role for
the evaluation of the function. We will denote by indeterminates these
variables.  We may increase the numbers of variables by keeping the
same number of indeterminates.
Let $f_n \in \F_n$ then for all $k > 0$, it exists $f_{n+k} \in
\F_{n+k}$, such that $f_n$ and $f_{n+k}$ share the same polynomial
form, $P(f_n) = P(f_{n+k})$. Furthermore 
$T(f_{2+k}) = T(f_2)*2^k$ (we repeat $2^k$ times the word $T(f_2)$).

For instance, the Boolean function with
two variables $f_2$ such that $T(f_2) = 0110$ has the polynomial form
$P = P(f_2) = x_1 \oplus x_2$ and $f_3$ such that $T(f_3) = 01100110$ 
satisfies $P(f_3) = x_1 \oplus x_2$.
\end{remark}
We were wondering if the Mobius transforms of function see as $n$
variables and $n+k$ are equals.
Let $f_n \in \F_n$ and $f_m \in \F_m$ such that $n<m$ and $P(f_n) = P(f_m)$. 
Does the following equality holds?
\begin{equation}
  \label{eq:equalityMob}
  \mu(f_n) = \mu(f_m)?
\end{equation}
We propose to look on a particular toy example.
\begin{example}
  \label{ex:equalityMob}
  Let  $f_1 =  x_1 \oplus  1 \in  \F_1$ be a boolean function with one
  indeterminate and $f_2 =  x_1 \oplus 1 \in \F_2$ be the same
  function see as a boolean function with two variables. Very
  simple  computations give  us that  $\mu(f_1) =  1$ and \mbox{$\mu(f_2) =
  x_2\oplus1$}.
\end{example}

\subsection{New insights of the Mobius transform}
\label{ssec:point_of_view}

The example~\ref{ex:equalityMob} provides a counterexample (\ref{eq:equalityMob}). 
Then we have to use the Mobius transform carefully if we manipulate the 
polynomial form of a Boolean function. 
In this setting,  we introduce the following notation.
\begin{notation}
  Let $f \in \F_n$ be a Boolean function with $n$ variables.
 We will write $\mu_n(f)$ instead of $\mu(f)$.
\end{notation}
 
Example~\ref{ex:equalityMob} implies  $\mu(x_1\oplus1) =
  \mu_1(x_1\oplus1) = 1$ and $\mu_2(x_1\oplus1) = x_2\oplus1$.
In other words,  the  Mobius  transform depends  on the variables, not on the indeterminates. 

We obtain the following result, which permits us to
manipulate the Mobius transform easily, hence this is one of key
ingredient.

The following theorem contains three key ingredients to manipulate properly  
the Mobius transform.

\begin{theorem}
  \label{prop:atomicMuEquality}
  Let $f\in \F_{n-1}$. We have
  \begin{equation}
    \label{eq:muNp1Fn}
    \mu_{n}(f) = (1\oplus x_{n})\mu_{n-1}(f),
  \end{equation}
  \begin{equation}
    \label{eq:muXNp1Fn}
    \mu_{n}(x_{n}f) = x_{n}\mu_{n-1}(f),
  \end{equation}
  \begin{equation}
    \label{eq:mu1pXNp1Fn}
    \mu_{n}((1\oplus x_{n})f) = \mu_{n-1}(f).
  \end{equation}
\end{theorem}
\begin{proof}
  It is straightforward to prove Equation~(\ref{eq:muNp1Fn}) from
  the definition of the Mobius transform. Indeed, we can see the
  computation of the Mobius transform as an interpolation
  problem. Since the coefficients of the monomials where occurs
  $x_{n+1}$ must be null, thus the statement.\\
The definition of the Mobius transform give us 
  \begin{eqnarray*}
    \mu_{n}(x_{n}f) & = & \mu_{n}(f) - \mu_{n-1}(f)\\
    & = & (1\oplus x_{n}) \mu_{n-1}(f) - \mu_{n-1}(f)\\
    & = & x_{n} \mu_{n-1}(f),
  \end{eqnarray*}
  which is exactly Equation~(\ref{eq:muXNp1Fn}).\\

  Let us to compute Equality~(\ref{eq:mu1pXNp1Fn}):
  \begin{eqnarray*}
    \mu_{n}((1\oplus x_{n})f) & = & \mu_{n}(f) \oplus \mu_{n}(x_{n}f)\\
    & = & (1 \oplus x_{n})\mu_{n-1}(f) \oplus x_{n} \mu_{n-1}(f)\\
    & = & \mu_{n-1}(f).
  \end{eqnarray*}
\end{proof}

The Reed-Muller decomposition is related to algebraic normal form,
whereas the Shannon one is related to truth table. Since the Mobius
transform allows us to switch from one to the other; it is natural to see
the relationship between the previous decompositions and the Mobius
transform.
\begin{proposition}
  \label{prop:decomposition1}
  \cite[Theorem 5]{coincidentFunc}.
  Let $f \in \F_n$ be a boolean function and $f_R^0, f_R^1 \in \F_{n-1}$
  be the terms of the Reed-Muller decomposition. Then
  $$
  \mu_n(f) = (1\oplus x_n)\mu_{n-1}(f_R^0) \oplus x_n \mu_{n-1}(f_R^1).
  $$
\end{proposition}
\begin{proof}
  \begin{eqnarray*}
    \mu_n(f) & = & \mu_n(f_R^0) \oplus \mu_n(x_nf_R^1)\\
    & = & (1\oplus x_n)\mu_{n-1}(f_R^0) \oplus x_n\mu_{n-1}(f_R^1).
  \end{eqnarray*}
\end{proof}
\begin{proposition}
  \label{prop:decomposition2}
  Let $f \in \F_n$ be a boolean function and $f_S^0, f_S^1 \in \F_{n-1}$
  be the terms of the Shannon decomposition. Then
  $$
  \mu_n(f) = \mu_{n-1}(f_S^0) \oplus x_n \mu_{n-1}(f_S^1).
  $$
\end{proposition}
\begin{proof}
  \begin{eqnarray*}
    \mu_n(f) & = & \mu_n\left((1\oplus x_n)f_S^0\right) \oplus \mu_n(x_nf_S^1)\\
    & = & (1 \oplus x_n) \mu_{n-1}\left(f_S^0 \right)
    \oplus x_n \mu_{n-1}\left( f_S^0\right)
    \oplus x_n \mu_{n-1}\left( f_S^1\right)\\
    & = & \mu_{n-1}(f_S^0) \oplus x_n \mu_{n-1}(f_S^1)
  \end{eqnarray*}
\end{proof}
We find again that the Reed-Muller decomposition of the Mobius transform
is the Mobius transform of the Shannon decomposition; and vice
versa.\\

The first step to try to make a link between the algebraic degree of a
boolean function and its hamming weight is given by the following
proposition.
\begin{proposition}
\label{prop:degre}
\cite[Theorem 7]{coincidentFunc}
 Let $f\in \F_n \setminus \{ \mathbf 0_n \}$. 
Then 
$$
\deg(f) + \deg(\mu_n(f)) \ge n.
$$
\end{proposition}
\begin{proof}
Let $f = f_R^0 \oplus x_n f_R^1$. By Proposition~\ref{prop:decomposition1},
\begin{eqnarray*}
\mu_n(f) & = & \mu_{n-1}(f_R^0) \oplus x_n (\mu_{n-1}(f_R^0) \oplus \mu_{n-1} (f_R^1))\\
   & = &   \mu_{n-1}(f_R^0) \oplus x_n (\mu_{n-1}(f_R^0 \oplus f_R^1)).
\end{eqnarray*}
Since $f \not= \mathbf 0_n, f_R^0$ and $f_R^1$ cannot be null in same time, then
it is easily seen that the property holds for $n = 1$. 

Let $n > 1$. Assume that the property holds for $n-1$.

\noindent Case 1) $\deg(f_R^0) \not= \deg(f_R^1)$. Then 
$$
\deg(f_R^0 \oplus f_R^1) = \max(\deg(f_R^0),\deg(f_R^1)) \le \deg(f).
 $$
Since $\deg(f_R^0) \ne \deg(f_R^1)$, then $\deg(x_n \mu_{n-1}\left(f_R^0 \oplus f_R^1 \right))\ge \deg(\mu_{n-1}\left( f_R^0 \oplus f_R^1 \right))$, so we deduce
\begin{eqnarray*}
  \deg(f) + \deg(\mu_n(f)) & \ge & \deg(f_R^0 \oplus f_R^1) + \deg(\mu_{n-1}(f_R^0 \oplus f_R^1)) + 1\\
  & \ge & n - 1 + 1 \hspace{2cm} \mbox{by hypothesis of recurrence}\\
  & \ge & n
\end{eqnarray*}

\noindent Case 2) Hence $\deg(f_R^0) = \deg(f_R^1)$; moreover $\deg(f) =
\deg(f_R^0) + 1$ and \mbox{$\deg(\mu_n(f)) \ge \deg(\mu_{n-1}(f_R^0))$}, and
by hypothesis of recurrence, \mbox{$\deg(f_R^0) + \deg(\mu_{n-1}(f_R^0)) \ge
n-1$}, hence $\deg(f) + \deg(\mu_n(f)) \ge n$.
\end{proof}
Always with the same intended target: draw a link between algebraic
degree and the Hamming weight of a boolean function; the following
result gives us the probability that the Mobius transform has degree
$n$.
\begin{proposition}
\label{prop:degren}
Let $f \in \F_n$ built uniformly at random. Then
$$
\begin{array}{l}
\Pr( \deg(\mu_n(f)) = n \; | \; \deg(f) = n) = \Pr( \deg(\mu_n(f)) = n \; | \; \deg(f) < n) \\
= \Pr(\deg(\mu_n(f)) = n) = \frac{1}{2}.
\end{array}
$$
\end{proposition}
\begin{proof}
One may easily check the property for $n=1$. Assume the property holds
for $n-1$.

Let $f = f_R^0 \oplus x_n f_R^1$ built uniformly at random. 
Then $\deg(f) = n$ if and only if $x_1\ldots x_n$ occurs in the ANF of $f$. 
Since it is the case for half of the Boolean functions, $\Pr( \deg(f) = n) = \frac{1}{2}$.
On the one hand $\deg(f) = n$ if and only if \mbox{$d(f_R^1) = n-1$}
and on the other hand, by Proposition~\ref{prop:decomposition1}, 
$\deg(\mu_n(f)) = n$ if and only if $\deg(\mu_{n-1}(f_R^0 \oplus f_R^1)) = n-1$. 
Hence
{\small $$
\Pr( \deg(\mu_n(f)) = n \; | \; \deg(f) = n) = \Pr(\deg(\mu_{n-1}(f_R^0 \oplus f_R^1)) = n-1 \; | \;  \deg(f_R^1) = n-1).
$$}
Moreover, $\deg(\mu_{n-1}(f_R^0 \oplus f_R^1)) = n-1$ if $\deg(\mu_{n-1}(f_R^0)) = n-1$ and 
\mbox{$\deg(\mu_{n-1}(f_R^1)) < n-1$} or  $\deg(\mu_{n-1}(f_R^0)) < n-1$ and 
$\deg(\mu_{n-1}(f_R^1)) = n-1$.

The degree of $\mu_{n-1}(f_R^0)$ and $\mu_{n-1}(f_R^1)$ are clearly independent 
and $\Pr(\deg(\mu_{n-1}(f_R^0)) = n-1) = \frac{1}{2}$.
It follows  
\[
\begin{array}{l}
\Pr(\deg(\mu_{n-1}(f_R^0)) = n-1 \mbox{ and } \deg(\mu_{n-1}(f_R^1)) < n-1  |  \deg(f_R^1) = n-1 )\\ 
 =  \Pr(\deg(\mu_{n-1}(f_R^0)) = n-1) \Pr(\deg(\mu_{n-1}(f_R^1)) < n-1 \; | \;  \deg(f_R^1) = n-1) \\
 =  \frac{1}{2} \cdot \frac{1}{2} = \frac{1}{4}.
\end{array}
\]
Similarly, we find 
$$
\Pr(\deg(\mu_{n-1}(f_R^0)) < n-1 \mbox{ and } 
\deg(\mu_{n-1}(f_R^1)) = n-1  |  \deg(f_R^1) = n-1 ) = \frac{1}{4}.
$$
Hence  $\Pr(\deg(\mu_{n-1}(f_R^0 \oplus f_R^1)) = n-1 \; | \;  \deg(f_R^1) = n-1) = \frac{1}{2}$.
\end{proof}

The following proposition can be view as a new decomposition related
to Mobius transform. 
\begin{proposition} \label{prop:muProduct}
\cite[Lemma 7]{coincidentFunc}.
Let $n \in \N$ and \mbox{$0< k < n$}. We will denote respectively  
the tuples \mbox{$x=(x_1,\ldots, x_n)$}, $y=(x_1,\ldots, x_k)$ and \mbox{$z=(x_{k+1}, \ldots, x_n)$}. 
Let $f_1(y) \in \F_k$, \mbox{$f_2(z) \in \F_{n-k}$} and $f(x) = f_1(y) \cdot f_2(z)$ 
and $g(x) = f_1(y) \oplus f_2(z)$. 
Then
\[
\begin{array}{lcl}
\mu_n(f(x)) & = & \mu_{n-k}(f_2(z)) \cdot \mu_{k}(f_1(y));\\
\mu_n(g(x)) & = & \prod_{i=k+1}^n (1 \oplus x_i) \mu_k((f_1(y)) \oplus \prod_{i=1}^k (1 \oplus x_i) \mu_{n-k} (f_2(z)).
\end{array}
\]
\end{proposition}
\begin{proof}
Let $u = (u_{k+1},\ldots,u_{n}) \in \K_2^k$. By applying (\ref{eq:muNp1Fn}) and 
(\ref{eq:muXNp1Fn}) of Proposition~\ref{prop:atomicMuEquality}, it follows
\[
\mu_{n-k}(z^u f_1(y)) = M_u(z) \cdot \mu_k(f_1(y)).
\]
Hence
\[
\begin{array}{lcl}
\mu_n(f(x)) & = & \bigoplus_{u \in \K_2^n} M_u(z) \cdot \mu_k(f_1(y))\\
         & = & \mu_{n-k} (\bigoplus_{u \in \K_2^n} z^u) \cdot \mu_k(f_1(y))\\
         & = &  \mu_{n-k}(f_2(z)) \cdot \mu_k(f_1(y)).
\end{array}
\]
First $\mu_n(g(x)) = \mu_n(f_1(y)) \oplus \mu_n(f_2(z))$. Furthermore 
(\ref{eq:muNp1Fn}) of Proposition~\ref{prop:atomicMuEquality} implies
$\mu_n(f_1(y)) = \prod_{i=k+1}^n (1 \oplus x_i) \mu_k((f_1(y))$ 
and \mbox{$\mu_n(f_2(z)) = \prod_{i=1}^k (1 \oplus x_i) \mu_{n-k} (f_2(z))$}.
\end{proof}

\subsection{Coincident functions}
\label{ssec:coincident_functions}
The notion of coincident function was previously introduced in
\cite{coincidentFunc}. A coincident function is a boolean function
which is equal to its Mobius transform.
\begin{definition}[Coincident function]
  Let $f\in \F_n$, $f$ is called
 {\em coincident} if and only if 
  $$
  f = \mu_n(f).
  $$
We will denote by $\C _n$ the set of coincident functions with $n$ variables.
\end{definition}
For this particular subset of boolean functions the monomial and the
associated minterms are directly related; this is the result of the
following proposition.  Let $f \in \mathcal F_n$, by
Proposition~\ref{involution}, $f \oplus \mu_n(f) \in \C_n$.  Here we
propose an improvement of \cite[Lemma 11]{coincidentFunc}.
\begin{proposition} \label{prop:coincidentMonomial}
Let $u \in \K_2^n$ and $h_u = x^u \oplus M_u$. 
Then $h_u$ is coincident and
\[
h_u = \bigoplus_{u \prec v} x^v = \bigoplus_{u \prec v} M_v.
\]
\end{proposition}
\begin{proof}
Since $\mu_n(x^u) = M_u$, $h_u = x^u \oplus \mu_n(x^u)$ is coincident.
By $(\ref{minterm2monomial})$, $M_u = \bigoplus_{u \preceq v} x^v$ hence 
$h_u =  \bigoplus_{u \prec v} x^v$.

\end{proof}
\begin{remark}
  We propose another point of view of the previous proposition: let
  \mbox{$f_k = \prod_{i=1}^{k}x_i$}, the function with all multiples
  of $f_k$, except $f_k$ itself, give a coincident function; that is
$$
f_k\left( 1 \oplus \prod_{i=k+1}^{n}\left( 1 \oplus x_i\right)\right).
$$
\end{remark}
Hence for particular coincident functions, it becomes trivial to
compute their hamming weight.
\begin{corollary}
Let $k = w_H(a)$. It follows $w_H(h_a) = 2^{n-k} -1$.
\end{corollary}

Coincident functions had a lower bound on its algebraic degree, this
is a trivial consequence from Proposition~\ref{prop:degre}.
\begin{proposition}
  \label{prop:coincidentWeight}
  \cite[Theorem 31]{coincidentFunc} Let $h \in \C_n \setminus \lbrace
  0 \rbrace$ be a coincident function with n variables. Then a lower
  bound on its algebraic degree is
  $$
  \deg(h)\ge \frac{n}{2}.
  $$
\end{proposition}
 \begin{proof}
   Since $h$ is a coincident function and Proposition~\ref{prop:degre}, we have that
   \begin{eqnarray*}
    \deg(h) + \deg(\mu(h)) & \ge & n\\
    2 \deg(h) & \ge & n\\
    \deg(h) & \ge & \frac{n}{2}.
   \end{eqnarray*}
 \end{proof}

\begin{notation}
Let $f \in \F_n$.
We define the operators $\varphi_n$ and $\C_n$ by
\[
\begin{array}{lcl}
\varphi_n(f) & = & f \oplus \mu_n(f) \\
\C_n(f) & = & \left\lbrace f' \in \F_n : \varphi_n(f') = \varphi_n(f)\right\rbrace.\\
\end{array}
\]
\end{notation}

\subsection{A construction of coincident functions}
\label{ssec:construction}

The following proposition gives us a simple way to build a coincident functions.
\begin{proposition} \label{prop:equivCoincident}
  \cite[Theorem 24]{coincidentFunc}
Let $h \in \C_n$, there exists a unique $g \in \mathcal F_{n-1}$ such that 
\[
\begin{array}{lcl}
h & = & g \oplus \mu_n(g)\\
  & = & \varphi_{n-1}(g) \oplus x_n \mu_{n-1}(g)\\
  & = & (1 \oplus x_n) \varphi_{n-1}(g) \oplus x_n g. 
\end{array}
\]
\end{proposition}
\begin{proof}
Let $h = h_R^0 \oplus x_n h_R^1 \in \C_{n-1}$.
Then $\mu_n(h) = \mu_{n-1}(h_R^0) \oplus x_n \mu_{n-1}(h_R^0 \oplus h_R^1)$. 
Since $h = \mu_n(h)$, it follows $h_R^0 \in \C_{n-1}$ and $h_R^1 = h_R^0 \oplus \mu_{n-1}(h_R^1)$, 
which implies $h_R^0 = \varphi_{n-1}(h_R^1)$.
Let $g = \mu_{n-1}(h_R^1)$, hence
\begin{eqnarray*}
  g  \oplus \mu_n(g)  & =  &  \mu_{n-1}(h_R^1) \oplus  (1 \oplus  x_n)
  h_R^1\\
  & = & \varphi_{n-1}(h_R^1) \oplus x_n h_R^1\\
  & = & h_R^0 \oplus x_n h_R^1 = h.
  \end{eqnarray*}
Since $\varphi_{n-1}(g) = h_R^0$, we also have 
\begin{eqnarray*}
  h & = & \varphi_{n-1}(g) \oplus x_n \mu_{n-1}(g)\\
  & = & (1 \oplus x_n) \varphi_{n-1}(g) \oplus x_n g.
  \end{eqnarray*}
\end{proof}
Since we have a one to one correspondence between the Boolean functions with $n-1$ variables and 
the coincident functions with $n$ variables, we trivially deduce that $card(\C_n) =
2^{2^{n-1}}$; see \cite[Theorem 18]{coincidentFunc}.
We also show that $\C_n$ has a vectorial space structure.
\begin{proposition}
The set $\C_n$ is a vectorial space of dimension $2^{n-1}$.
\end{proposition}
\begin{proof} Let $h_1$ and $h_2 \in \C_n$. We consider $f_1$ and $f_2 \in \F_2^n$ such that 
$h_1 = \varphi_n(f_1)$ and $h_2 = \varphi_n(f_2)$. Let $h = h_1 \oplus h_2$. 
\[
h = g_1 \oplus \mu_n(g_1) \oplus g_2 \oplus \mu_n(g_2) =  (g_1 \oplus g_2) \oplus \mu_n(g_1 \oplus g_2),
\]
hence $h = \varphi_n(f_1 \oplus f_2)$.
\end{proof}
\begin{corollary} (Random generation of a coincident function) 
Let $uniform(E)$ a function which returns with the uniform distribution 
an element of a finite set $E$. The following algorithm returns with the uniform distribution a coincident function from $\C_n$. 
\begin{algorithm}[h]
  \caption{\label{algo:randomCoincident} 
    Uniform random generation of a coincident function}
  \KwIn{The integer $n$, the number of variables.}
  \KwOut{$h \in \C_n$ a coincident function with $n$ variables.}
  \BlankLine
\Begin{$g \gets$ uniform($\F_{n-1})$\\
$f \gets \mu_n(g)$\\
$h \gets g \oplus f$
}
  \KwRet{$h$}
\end{algorithm}
\end{corollary}

\begin{proposition}
Let $h \in C_n$ be a coincident function, then $h(0_n) = 0$.
\end{proposition} 
\begin{proof}
Let $f$ such that $h = \varphi_n (f)$, by Proposition~\ref{prop:firstMobius},
$\mu(f)(0_n) = f(0_n)$, then \mbox{$h(0_n) = 0$}.
\end{proof}

Here is a strong connection between the Mobius transform of a boolean function and this one of its Reed-Muller decomposition.
\begin{proposition}\label{recursive2}
  Let $f = f^0_R \oplus x_n f_R^1 \in \F_n$ and $h^0 = \varphi_{n-1}(f_R^0)$, $h^1 = \varphi_{n-1}(f_R^1)$ then
  $$
  \varphi_n(f) = h^0 + x_n(h^0\oplus h^1 \oplus f_R^0).
  $$
\end{proposition}
\begin{proof}
  Let $h = \varphi_n(f)$. Let us to compute $h$
  \begin{eqnarray*}
    h & = & f \oplus \mu_n(f)\\
    & = & f_R^0 \oplus x_n f_R^1 \oplus (1+x_n)\mu_{n-1}(f_R^1) \oplus x_n \mu_{n-1}(f_R^1)\\
    & = & h^0 \oplus x_n(h^1 \oplus  \mu_{n-1}(f_R^0))\\
    & = & h^0 \oplus x_n(h^0 \oplus h^1 \oplus f_R^0).
  \end{eqnarray*}
\end{proof}
\begin{remark}
Proposition~\ref{prop:equivCoincident} and~\ref{recursive2} imply
$h^0 = \varphi_{n-1}(g)$ and $g = h^1 \oplus f_R^0$, {\it ie} \mbox{$f_R^0 \in \C_{n-1}(g)$}.
\end{remark}

\begin{remark}
  Let $h \in \F_n$, then $h$ is coincident if and only if $h \in \C_n(0)$.
\end{remark}
\begin{corollary}
  \label{prop:additiveClass}
  Let $g_1, g_2 \in \F_{n-1}$, $f_1 \in \C_n(g_1)$ and $f_2 \in
  \C_n(g_2)$. Then 
  $$
  f_1 \oplus f_2 \in \C_n(g_1 \oplus g_2).
  $$
\end{corollary}

From the linearity, the Mobius transform of the sum of function is the
sum of the Mobius transforms. We propose to look on the
multiplicativity.
\begin{proposition}
\label{productCoincident}
  \cite[Theorem 7]{coincidentFunc}
Let $n \in \N$ and $0< k < n$. We will denote respectively
the   tuples   $x=(x_1,\ldots,   x_n)$,  $y=(x_1,\ldots,   x_k)$   and
\mbox{$z=(x_{k+1},  \ldots,  x_n)$}.  Let  $f_1(y)  \in  \F_k$,  $f_2(z)  \in
\F_{n-k}$ and $f(x) = f_1(y)  \cdot f_2(z)$. Then $f$ is coincident if
and only if $f_1$ and $f_2$ are.
\end{proposition}
\begin{proof}
This is a direct application of Proposition~\ref{prop:muProduct}:
\[
\mu_n(f(x)) = \mu_{n-k}(f_2(z)) \cdot \mu_k(f_1(y)).
\]
Clearly if $f_1$ and $f_2$ are coincident functions, then 
\mbox{$\mu_n(f(x)) = f_2(z) \cdot f_1(y) = f(z)$} and $f$ is a coincident function.
Conversely, $\mu_n(f(x)) = f(z)$ implies \mbox{$\mu_{n-k}(f_2(z)) = f_2(z)$}
and $\mu_k(f_1(y)) = f_1(y)$.

\end{proof}

We exhibit now some constructions of coincident functions available for any number 
of variables.
\begin{proposition}
  \label{prop:atomicCoincidentFct}
  Let $n \in \mathbb{N}^*$ be a positive integer, the boolean functions:
  \begin{enumerate}
  \item \label{NulleFct} $\mu_n(0_n) = 0_n$, 
  \item \label{MonomialPrdt} $\prod_{i=1}^n x_i$,
  \item \label{SumMonomial} $1\oplus  \prod_{i=1}^n(1\oplus  x_i)$,
  \end{enumerate}
  are coincident.
\end{proposition}
\begin{proof} Let us to prove the three assertions.\\
The definition of Mobius transform gives us directly \ref{NulleFct}.\\
For the next assertion \ref{MonomialPrdt}, we have $\prod_{i=1}^n x_i =
x^{\mathbf 1} = M_{\mathbf 1}$, we conclude by applying
Proposition~\ref{prop:muminterm}.\\
Concerning the assertion\ref{SumMonomial};
$\prod_{i=1}^n(1\oplus  x_i) = M_{\mathbf 0}$ and $1 = x^{\mathbf 0}$, then  
\mbox{$\prod_{i=1}^n(1\oplus  x_i) = \mu_n(1)$} and $1 = \mu_n(\prod_{i=1}^n(1\oplus  x_i))$. 
Hence $1\oplus  \prod_{i=1}^n(1\oplus  x_i) = \varphi_n(1)$.
\end{proof}
Since the Mobius transform of a coincident function provides a connection between the minterms and the monomials, 
we deduce the following corollary.
\begin{corollary}
  \label{coro:weightCoincident}
  Let $h$ be a coincident function, then $w(h) = N(h)$,
  where $N(h)$ gives the number of monomials of the $h$.
\end{corollary}
We introduce now the dual of a coincident function.
\begin{definition} 
Let $h \in \C_n$. The dual of $h$ is the coincident function 
$\overline{h}^* = h \oplus \varphi_n(1)$.
\end{definition}

\begin{proposition}\label{parityCoincident}
We have a one to one correspondence between coincident functions with odd 
Hamming weight and even Hamming weight.
\end{proposition}
\begin{proof}
A Boolean function has Hamming weight odd if and only if $\prod_{i=1}^n x_i$ occurs in its ANF.
Let $h \in \C_n$. It follows $A(f) = T(f) = t_1 \ldots t_{2^n}$, 
and $w_H(h)$ odd if and only if  $t_{2^n} = 1$. Hence we a have partition of 
$\C_n = (\C_n^o,\C_n^e)$, where $\C_n^o$ (resp. $\C_n^e$) are coincident functions of Hamming weight odd (resp. even) 
and a one to one correspondence between $\C_n^o$ and $\C_n^e$ with $h' = h \oplus \prod_{i=1}^n x_i$, 
$\prod_{i=1}^n x_i$ is the coincident function which changes the parity of a coincident function.
\end{proof}
Thanks to the previous propositions, we propose to exhibit different
coincident functions of any number of variables.

\begin{proposition} Table of some coincident functions\\
The following words codes table (or equivalently ANF) of coincident functions:
\begin{enumerate}
\item $\mathbf 0_n \Longleftrightarrow 0\ldots0$;
\item $\prod_{i=1}^n x_i \Longleftrightarrow 0\ldots01$;
\item $1 \oplus \prod_{i=1}^n(1\oplus  x_i) \Longleftrightarrow 01\ldots1$;
\item $1 \oplus \prod_{i=1}^n(1\oplus  x_i) \oplus \prod_{i=1}^n x_i =
  \overline{\prod_{i=1}^n x_i}^* \Longleftrightarrow 01\ldots 10$;
\item $\forall u \in \K_2^{n},\ \bigoplus_{u \prec v} x^v$.
\end{enumerate}

\end{proposition}

\section{Constructions of coincident functions based on the Boolean lattice}
\label{sec:lattices}

\subsection{Boolean lattice properties of coincident functions}
The set of valuations in $n$ variables forms a Boolean lattice with the partial order
$\preceq$ already defined. The process of placing the minterms (or the monomials) 
which appears in a Boolean function over this lattice may be useful for some studies.
In \cite{guillot} PHD Thesis Boolean lattice is already considered for the study of Mobius transform. 

We first provide a new caracterisation of coincident functions.
Let $n \in \N$, we define the complete Boolean lattice $\mathcal L_n = (\K_2^n,\preceq)$ 
such that, for any $u = (u_1,\ldots, u_n)$ and $v = (v_1,\ldots, v_n) \in \K_2^n$, 
the supremum and the infimum are defined by
\[
\begin{array}{lcl}
sup(u, v) & = & u \vee v   =  (u_1 \vee v_1, \ldots, u_n \vee v_n);\\
inf(u,v) & = & u \wedge v  =  (u_1 \wedge v_1, \ldots, u_n \wedge v_n).
\end{array}
\]

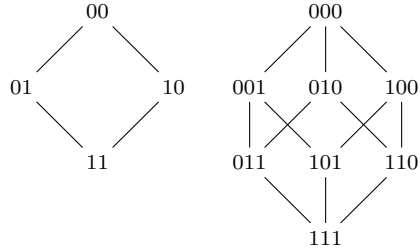
\begin{figure}[h]
  \centering
  \begin{tikzpicture}
    \node (B) at (0, 0) {$11$};
    \node (G) at (-1, 1) {$01$};
    \node (H) at (0, 2) {$00$};
    \node (D) at (1, 1) {$10$};
    
    \draw (B) -- (G);
    \draw (G) -- (H);
    \draw (H) -- (D);
    \draw (D) -- (B);

    \begin{scope}[xshift=3cm]
      \node (B) at (0, 0) {$101$};
      \node (G) at (-1, 1) {$001$};
      \node (H) at (0, 2) {$000$};
      \node (D) at (1, 1) {$100$};
      
      \draw (B) -- (G);
      \draw (G) -- (H);
      \draw (H) -- (D);
      \draw (D) -- (B);

      \begin{scope}[yshift=-1cm]
        \node (Bb) at (0, 0) {$111$};
        \node (Gb) at (-1, 1) {$011$};
        \node (Hb) at (0, 2) {$010$};
        \node (Db) at (1, 1) {$110$};
        
        \draw (Bb) -- (Gb);
        \draw (Gb) -- (Hb);
        \draw (Hb) -- (Db);
        \draw (Db) -- (Bb);

        \draw (Bb) -- (B);
        \draw (Gb) -- (G);
        \draw (Hb) -- (H);
        \draw (Db) -- (D);
      \end{scope}

    \end{scope}
  \end{tikzpicture}
  \caption{\label{fig:lattice_rep}
  The lattices $\mathcal L_2$ and $\mathcal L_3$.}
\end{figure}

\begin{figure}[h]
  \centering
  \begin{tikzpicture}
    \node (B) at (0, 0) {$1_k$};
    \node (G) at (-1, 1) {};
    \node (H) at (0, 2) {$0_k$};
    \node (D) at (1, 1) {};
    
    \draw (B) -- (G);
    \draw (G) -- (H);
    \draw (H) -- (D);
    \draw (D) -- (B);

    \draw (0, -1) node {The Boolean lattice $\mathcal{L}_k$};
    
    \begin{scope}[xshift=3cm]
      \node (B) at (0, 0) {$u \vee v$};
      \node (G) at (-1, 1) {$u$};
      \node (H) at (0, 2) {$u \wedge v$};
      \node (D) at (1, 1) {$v$};
      
      \draw (B) -- (G);
      \draw (G) -- (H);
      \draw (H) -- (D);
      \draw (D) -- (B);

      \draw (0, -1) node {The Boolean lattice $\mathcal{L}(u, v)$};
    
    \end{scope}
  \end{tikzpicture}
  \caption{\label{fig:lattice_iso}
    Isomorphism between $\mathcal{L}_k$ and $\mathcal{L}(u,v)$.}
\end{figure}

The lattice is complete in the sense that any sublattice  $(\mathcal U, \preceq)$ has 
the supremum  $\bigvee_{u \in \mathcal U} u$ and the infimum $\bigwedge_{u \in \mathcal U} u$.
$\mathcal L_n$ is also called the $n$-cube (\cite{palmer1992}). 
Let $u$ and $v \in \K_2^n$ and $k = d_h(u \vee v, u \wedge v)$. 
Then \mbox{$\mathcal{L}(u,v) = \{ w \in \K_2^n | u \preceq w \preceq v\}$} is isomorph to 
$\mathcal L_k$ (we remove the identical components in $u$ and $v$),
see Fig.~\ref{fig:lattice_iso}.

\begin{definition}
Let $f \in \F_n$, we define $\mathcal L_n(f) = (\mathcal U, \preceq)$ 
the sublattice of $\mathcal L_n$ by $\mathcal U = support(f) =  \{ u \in \K_2^n | f(u) = 1\}$. 
\end{definition}
Then $f$ can be viewed as a $2$ coloring of the $\mathcal L_n$.

Let $a = (a_1,\ldots,a_n)$ and $f \in \F_n$,
${\displaystyle f(a) = \bigoplus_{u \in \K_2^n} \mu_n(f) (u) a^u}$.
Since $a^u = 0$ for any $u$ such that $u \not\preceq a$ and $a^u = 1$, it follows
\begin{equation} \label{mobiusLattice}
f(a) = \bigoplus_{u \preceq a} \mu_n(f) (u).
\end{equation}
\begin{proposition} \label{coincidentLattice}
  \cite[Theorem 23]{coincidentFunc}
  Let $f \in \F_n$ be a boolean function with $n$ variables. Then
$f$ is coincident if and only if \begin{equation} \label{prop}
\bigoplus_{v \prec u} f(v) = 0, \mbox{ for any } u \in \K_2^n.
\end{equation}
In other words, for each $u \in \K_2^n$, we have an even number of $v \prec u$ such that $f(v) = 1$.
\end{proposition}
\begin{proof}
Let $f$ be a coincident function, then
$f(u) = \bigoplus_{v \preceq u} f(v) = f(u) \oplus \bigoplus_{v \prec u} f(v)$.
\end{proof}
\begin{proposition} \label{basisCoincident}
\label{coincidentBasis}
Recall that $h_a =x^a \oplus M_a$, for any $a \in \K_2^n$. Then
$(h_a)_{a \in \K_2^n,\ a_n = 0} $ forms a basis of $\mathcal C_n$.
\end{proposition}
\begin{proof} 
Let $h \in \mathcal C_n$ and $g \in \mathcal F_{n-1}$ such that $h = g \oplus \mu_n(g)$. 
Let $\mathcal U \subset \K_2^{n}$ such that $g = \bigoplus_{a \in \mathcal U} x^a$. 
Then $a_n = 0$ for any $a \in \mathcal U$ and
\[
\begin{array}{lcl}
h & = & \bigoplus_{a \in \mathcal U} x^a \oplus \mu_n(  \bigoplus_{a \in \mathcal U} x^a )\\
& = &  \bigoplus_{a \in \mathcal U} x^a \oplus \bigoplus_{a \in \mathcal U} M_a\\
& = &  \bigoplus_{a \in \mathcal U} h_a.
\end{array}
\]
\end{proof}
The following proposition gives a link between a Boolean function $f$ 
and the corresponding coincident function $\varphi_n(f) = f \oplus \mu_n(f)$.
\begin{proposition}\label{prop:definitionCoincidentLattice}
Let $f$ a Boolean function, $\mathcal U = support(f)$ and $h =
\varphi_n(f)$, then $h(a) = 1$ if and only if there is an odd number
of $u \in \mathcal U$ such that $u \prec a$.
\end{proposition}
\begin{proof}
Let $u \in \mathcal U$, $h_u = x^u \oplus M_u = \bigoplus_{u \prec v} M_v$.
Hence $h =  \bigoplus_{u \in \mathcal U} (\bigoplus_{u \prec v} M_v)$.
Let $\mathcal U_{\prec a} = \{u \in \mathcal U, u \prec a\}$, for any  $a \in \K_2^n$. 
Since $h(a) = \bigoplus_{u \in \mathcal U_{\prec a}} 1$, 
$h(a) = 1$ if and only if the cardinality of $\mathcal U_{\prec a}$ is odd.
 \end{proof}

\begin{proposition} \label{prop:coincidentDegree1}
The Boolean function with $n$ variables $x_1 \oplus \ldots \oplus x_n$ 
is coincident.
\end{proposition}
\begin{proof}
Let $f = \prod_{i=1}^n (1 \oplus x_i)$. $f(u) = 1$ if and only if $u = 0_n$.
Let $0_n = (0,\ldots,0)$ and $u^j = (u_1^j,\ldots, u_n^j)$ for any $j \in \{1,\ldots, n\}$, where 
$u_i^j = 1$ if and only if $i = j$; hence $x^{u^j} = x_i$. We check Proposition~\ref{prop:definitionCoincidentLattice}. 
There is no $u \prec 0_n$ and for any $u_i^j$, the unique $u \prec x^{u^j}$ is $0_n$ and $f(0_n) = 1$. 
\end{proof}

\begin{proposition}\label{prop:monomialDegree1}
Let $h = \varphi_n(f) \in \mathcal C_n$. Either the $ANF(h)$ contains all the terms $x_1, \ldots, x_n$ 
either it contains none of these terms. 
\end{proposition}
\begin{proof}
The proof is similar that Proposition~\ref{prop:coincidentDegree1}.
If $f(0_n) = 1$ (resp. $0$), then all the monomials $x_i$ has an odd (resp. even) number of $u \prec u^j$  
such that $f(u) = 1$.
\end{proof}

\subsection{Monotonic coincident functions}
Monotonic Boolean functions are commonly involved as instances of
constraint satisfaction problems like the NP-complete SAT or $3$-SAT
problems, \cite{creignou2001}. Indeed these problems are monotonic in
the sense where an instance of a set of constraints $\mathcal C$ is
also an instance of any subset of $\mathcal C$.
 
We exhibit in this part $2^{n+1}$ monotonic coincident functions. 
We also provide a general characterization of the class of monotonic
coincident functions in order to build more such functions.

\begin{definition}
A Boolean function $f$ is monotonic if for any $u \in \K_2^n$ such 
that $f(u) = 1$ we have $f(v)= 1$ for any $v \prec u$. 
\end{definition}
As far as we know, the number of such functions on $n$ variables is
known as the Dedekind number of $n$ and the exact values are known
only for $n \le 8$.

\begin{proposition}
A Boolean function $f$ is monotonic if there exists $Inf(f) \subset \K_2^n$ 
which satisfies $f(v) = 1$ if and only if there exists $u \in Inf(f)$ with $u \preceq v$.
\end{proposition}
\begin{proof}
We define $Inf(f) = \{ u \in \K_2^n \; | \; f(u) = 1 \mbox{ and } f(v) = 0 \mbox{ for any } v \prec u \}$.
Assume that $f$ is monotonic. Let $v \notin Inf(f)$ and $f(v)= 1$. Then there exists 
a unique $u \in Inf(f)$ such that $u \prec v$. 
Conversely assume that $f(v) = 1$ for any $v \preceq u$, for some $u \in Inf(f)$, 
then $f$ is clearly monotonic.
\end{proof}

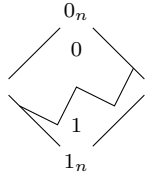
\begin{figure}[h]
  \centering
  \begin{tikzpicture}
    \node (B) at (0, 0) {$1_n$};
    \node (G) at (-1, 1) {};
    \node (H) at (0, 2) {$0_n$};
    \node (D) at (1, 1) {};
    
    \draw (B) -- (G);
    \draw (G) -- (H);
    \draw (H) -- (D);
    \draw (D) -- (B);
    
    \draw (-0.75, 0.75) -- (-0.25, 0.5) -- (0, 1) -- (0.5, 0.75) --
    (0.75, 1.25);

    \node at (0,1.5) {$0$};
    \node at (0,0.5) {$1$};

    \end{tikzpicture}
  \caption{\label{fig:lattice_mono}
  Lattice for monotonic boolean functions.}
\end{figure}

Since $f = \bigoplus_{u \in \K_2^n} \beta_u M_u$, $f$ is monotonic 
when for any $u$ such that $\beta_u = 1$, $\beta_v = 1$ for any 
$v \prec u$. 

For example, $x^u = \bigoplus_{u \preceq v} M_v$ are monotonic functions 
and $h_u = \bigoplus_{u \prec v} M_v$ are monotonic coincident functions.

\begin{proposition}
Let $u \in \K_2^n$ and $\overline{u} = (u_1\oplus 1,\ldots, u_n \oplus 1)$.
Then $f_u = h_u \oplus h_{\overline{u}} \oplus x_1\ldots x_n$ is a monotonic coincident function.
\end{proposition} 
\begin{proof}
Let $v \in \K_2^n$. If $h_u(v) = h_{\overline{u}}(v) = 1$ then $u \prec v$ 
and $\overline{u} \prec v$, hence \mbox{$v = (1,\ldots, 1)$} and $f_u(v) = 1$.
Assume that $v \not= (1,\ldots, 1)$, if $h_u(v) = 1$ or $h_{\overline{u}}(v) = 1$ then 
\mbox{$f_u(w) = 1$}, for any $v \preceq  w$.  
\end{proof}
We have yet exhibit $2^{n+1}$ monotonic coincident functions but other constructions over the Boolean lattice could be performed.

\subsection{Construction of the class of coincident symmetric Boolean functions}
Symmetric Boolean functions have good implementation since the number
of required gates is linear in the number of
variables. In \cite{canteaut2005}, the authors proposed an extensive study combined
with cryptographic parameters like degree, correlation-immunity,
non-linearity. We present here an algorithm to generate all the 
$\displaystyle{2^{\lfloor \frac{n}{2}\rfloor+1}}$ coincident symmetric functions.

\begin{definition}[Symmetric (Boolean) functions]
Let $k \le n$, $\Sigma_k^n$ will denote the Boolean function with $n$ variables which is the sum 
of monomials of degree $k$. 
A symmetric function $f$ of $\mathcal F_n$ is defined by 
\[
f = \sum_{k=0}^n \lambda_k \Sigma_k^n,
\]
where $(\lambda_0, \ldots, \lambda_n) \in F_2^n$ satisfies
\[
\lambda_i = 
\left\{
\begin{array}{lcl}
 1, & ~ &\mbox{ if all the monomials of degree } i \mbox{ occur in the ANF of f;}\\ 
 0, & ~ & \mbox{ otherwise.}
\end{array}
\right.
\]
We will note $\lambda(f) = (\lambda_0, \ldots, \lambda_n)$.

Since a symmetric function is invariant by permutation of the variables, we have another definition 
of a symmetric function over the valuations.
Let $a = (a_1,\ldots a_n) \in \F_2^n$,
\[
f (a_1,\ldots, a_n) = f (a_{\sigma (1)},\ldots, a_{\sigma(n)}), 
\]
for any permutation $\sigma$ of $\{ 1, \ldots n \}$.

Then the value of $f(a)$ only depends of the weight of the valuation $a$.
Let $v(f) = (v_0, \ldots, v_n)$, 
where $v_k = f (a)$, for any 
$a \in \F_2^n$ of weight $k$.
\end{definition}
Since a symmetric Boolean function may be defined by fixing \mbox{$\lambda(f) = (\lambda_0, \ldots, \lambda_n)$}
or $v(f_n) = (v_0, \ldots, v_n)$, we have $2^{n+1}$ symmetric Boolean functions.
Furthermore, $d(f)$ is the largest $i$ such that $\lambda_i = 1$ 
and $w_H(f)$ is $\sum_{i = 0}^n \lambda_i \; { n \choose i}$. 

Let $f = f_R^0 \oplus x_n f_R^1$ be a symmetric function with $n$ variables, 
where $f_R^0$ and $f_R^1$ are Boolean functions with $n-1$ variables.

For $i \in \{1,\ldots, n-1 \}$, the ANF of $f$ contains all the monomials of degree $i$ if and only if 
the ANF of $f_R^0$ contains all the monomials of degree $i$. Thus $f_R^0$ is a symmetric function.
Furthermore for $i \in \{0,\ldots, n-1 \}$ the ANF of $f$ contains all the monomials of degree $i+1$ if and only if 
the ANF of $f_R^1$ contains all the monomials of degree $i$. Hence $f_R^1$ is also a symmetric function.

We define 
\[
\begin{array}{lcl}
\lambda(f)   & = &  (\lambda_0, \ldots, \lambda_n);\\
\lambda(f_R^0) & = & (\lambda_0^0, \ldots, \lambda_{n-1}^0);\\
\lambda(f_R^1) & = & (\lambda_0^1, \ldots, \lambda_{n-1}^1).\\
\end{array}
\] 
We have $\lambda_i = \lambda_i^0 = \lambda_{i-1}^1$, for any $i \in \{ 1,\ldots, n-1\}$, 
$\lambda_n = \lambda_{n-1}^1$ and $\lambda_n= 1$ if and only if the ANF of $f$ 
contains the monomial $x_1\ldots x_n$.

\begin{definition}[Luca's coefficients]
Let $k$, $j \in \N$ and $p(k,j) = {\displaystyle { k \choose j} } \mbox{ mod } 2$.
\end{definition}
\begin{proposition}
Let $f$ be a symmetric Boolean function with $n$ variables, 
\mbox{$v(f) = ( v_0, \ldots, v_n)$} and 
$\lambda(f) = ( \lambda_0,\ldots, \lambda_n )$.
The following system holds
\[
v_j = \sum_{k=0}^j \lambda_k \; p(k,j). 
\]
\end{proposition}
\begin{proof}
Clearly $v_0 = \lambda_0$. 
Let $j \in \{ 0, \ldots, n\}$ and $a = (a_1,\ldots, a_n) \in \F_2^n$ such that $w_H(a) = j$. 
By (\ref{mobiusLattice}),
\[
f(a) = \bigoplus_{u \preceq a} \mu_n(f) (u) a^u.
\]
For $k \le j$, we have ${ k \choose j}$ $b \in \K_2^n$ such that 
$w_H(b) = k$ and $b \preceq a$. We have to deal with two cases
\begin{enumerate}
\item $\lambda_k = 1$ and $\mu_n(f) (b) = 1$ for any $b \preceq a$ such that $w_H(b) = k$.
\item $\lambda_k = 0$ and $\mu_n(f) (b) = 0$ for any $b \preceq a$ such that $w_H(b) = k$.
\end{enumerate}
It follows
\[
f(a) =  \bigoplus_{k=0}^j \lambda_k { k \choose j}.
\]
\end{proof}

\begin{notation}
Let $k = \sum_{i \in \N} k_i \; 2^i$ and $j = \sum_{i \in \N} j_i \; 2^i$ the $2$-adic representation 
of $k$ and respectively $j$.
We will write $j \preceq k$ when  $j_i = 1$ implies $k_i = 1$, for any $i \in \N$.
\end{notation}
By Lucas' Theorem 
$p(k,j) = 1 \mbox{ if and only if } j \preceq k$ and by definition of Mobius transform, 
\[
v(f) = \lambda(\mu_n(f)) \mbox{ and } \lambda(f) = v(\mu_n(f)).
\]
\begin{proposition}\label{musymmetric}
With the previous notations, we have 
$\lambda(\mu_n(\Sigma_n^k)) = (v_0^k,\ldots,v_n^k)$, where  
\begin{equation}\label{equationv}
\left\{
\begin{array}{lcl}
v^k_j & = & 0, \mbox{ for } j < k\\ 
v^k_k & = & 1\\
v^k_{j} & = & p(j,k), \mbox{ for } k < j \le n.\\
\end{array}
\right.
\end{equation}
\end{proposition}
We have already seen the symmetric coincident functions 
$1 \oplus x_1 \oplus \ldots \oplus x_n$ (Proposition~\ref{prop:coincidentDegree1}) 
which corresponds to $v(f) = (0,1,0,\ldots, 0)$.
We are looking for the whole class of such functions.

Recall that $h(\Sigma_k^n) =\Sigma_k^n \oplus \mu_n(\Sigma_n^k)$, 
$\lambda(h(\Sigma_n^k)) = (w_0^k,\ldots,w_n^k)$, 
where $w_i^k = v_i^k$, for any $i \not= k$ and $w_k^k = 0$.

Since a sum of symmetric functions still a symmetric function, a sum of symmetric coincident functions 
still a symmetric coincident function.
Hence $\mathcal{SC}_n$ the set of symmetric coincident function is generated by the $h(\Sigma_k^n)$ 
is a vector space of dimension $2^l$, for some $l \in n+1$. 
Remark that for some $k_1$ and $k_2 \in \{ 0,\ldots, n\}$, $k_1 < k_2$, we may have 
$h(\Sigma_n^{k_1}) =h(\Sigma_n^{k_2})$.

Let $CS_n$ the set of coincident symmetric Boolean functions. 
Since the sum of two coincident functions is a coincident function and 
the sum of two symmetric functions is a symmetric function, $CS_n$ is a vector space.
\begin{proposition}
\[
|CS_n| = \displaystyle{2^{\lfloor \frac{n}{2}\rfloor+1}}.
\]
\end{proposition}
\begin{proof} 
We show that 
\[
\left\{
\begin{array}{lcl}
|CS_1| & = & 2\\ 
|CS_n| & = & |CS_{n-1}| \mbox{ if } n \mbox{ is odd }\\
& = & 2 |CS_{n-1}| \mbox{ if } n \mbox{ is even.}
\end{array}
\right.
\]
For $n = 1$, $f(x_1) = x_1$ is the unique coincident symmetric Boolean function different from ${\mathbf 0}_1$, 
$\lambda(f) = (0,1)$ and $\lambda({\mathbf 0}_1) = (0,0)$.
 
Let $n \ge 1$ and $f = f_R^0 \oplus x_n f_R^1$ be a symmetric function with $n$ variables, 
where $f_R^0$ and $f_R^1$ are Boolean functions with $n-1$ variables.

Let $\lambda(f) = (\lambda_0, \ldots, \lambda_n)$ and $v(f) = (v_0,\ldots,v_n)$, 
$f$ is coincident if and only if $\lambda(f) = v(f)$.

Let $j$ be any element of $\{ 0,\ldots,n\}$.  
By implying (\ref{equationv}) we obtain
\[
\begin{array}{lcl}
v_j & = & \lambda_j \oplus \Big( \bigoplus_{k=0}^n \lambda_k \; v_j \Big),\\
& = & \lambda_j \oplus \Big( \bigoplus_{k < j} \lambda_k \; p(j,k) \Big).
\end{array}
\]
Then $f$ is coincident if and only if
\begin{equation} \label{conditionSC}
\bigoplus_{k < j} \lambda_k \; p(j,k) = 0, \mbox{ for any } j \in \{ 0,\ldots,n\}.
\end{equation}
We have seen that $\lambda(f_R^0) = (\lambda_0, \ldots, \lambda_{n-1})$. 
Furthermore
\[
\begin{array}{lcl}
\mu_n(f) & = & \mu_n(f_R^0) \oplus \mu_n(f_R^1)\\
         & = & (1 \oplus x_n) \mu_{n-1} (f_R^0) \oplus x_n \mu_{n-1} (f_R^1)\\
         & = & \mu_{n-1}(f_R^0) \oplus x_n \big( \mu_{n-1}(f_R^0)\oplus \mu_{n-1}(f_R^1) \big).
\end{array}
\] 
Then $f$ is coincident if and only if
$$
\left\{
\begin{array}{lcl}
\mu_{n-1}(f_R^0) & = & f_R^0\\
\mu_{n-1}(f_R^1) & = & f_R^0 \oplus f_R^1.
\end{array}
\right.
$$
The first equation implies that $f_R^0$ is a symmetric coincident function and it just 
remains to check the last equation of (\ref{conditionSC})
\[
\bigoplus_{k < n} \lambda_k \; p(n,k) = 0.
\]
\paragraph{Case $n$ is even}~\\
$p(n,1) = 0$ and it is easily seen that $p(n,k) = 0$ for $k$ odd and $p(n,k) = p(n-2,k-2)$, for $k$ even. 
Then we have to check 
\[
\bigoplus_{k < n-2} \lambda_k \; p(n-2,k) = 0,
\]
which is already satisfies by $f_R^0$.
Hence we may choose $\lambda_n = 0$ or $1$ and \mbox{$CS_n = 2 CS_{n-1}$}.

\paragraph{Case $n$ is odd}~\\ 
Since $p(n,1) = 1$, $\lambda_0 \; p(n,n) \oplus \lambda_2 \; p(n,n-2)
\ldots \oplus \lambda_{n-2} \; p(n,2) = \lambda_{n-1}$.
Then we may chose $\lambda_n = 0$ or $1$, but we have a unique possibility for  $\lambda_{n-1}$.
By the previous case, we know that half of the symmetric coincident function $f_R^0$ 
satisfy $\lambda_{n-1} = 0$. Clearly $|CS_n| = |CS_{n-1}|$. 
\end{proof}
\begin{proposition}Enumeration\\
The Algorithm~\ref{algo:enumSymmetric} provides an enumeration of coincident symmetric functions.
\begin{algorithm}[h]
  \caption{\label{algo:enumSymmetric}
    Enumeration}
  \KwIn{The number of variables: $n \in \N$.}
  \KwOut{Enumeration of coincident symmetric functions with $n$ variables.}
  \BlankLine
\If { $n = 1$}{
   enumerate $(0,0)$\\
   enumerate $(0,1)$\\
}
\Else{
\If {$n$ even}{
  \For{ $(\lambda_0,\ldots, \lambda_{n-1})$ of Enumeration($n-1$)}{
  enumerate $(\lambda_0,\ldots, \lambda_{n-1},0)$ \\
  enumerate $(\lambda_0,\ldots, \lambda_{n-1},1)$
}
}
\Else {
\For{ $(\lambda_0,\ldots, \lambda_{n-1})$ of Enum($n-1$)}{
\If{ $\lambda_0 \; p(n,n) \oplus \lambda_2 \; p(n,n-2) \ldots \oplus \lambda_{n-2} \; p(n,2) = \lambda_{n-1}$}{
 enumerate $(\lambda_0,\ldots, \lambda_{n-1},0)$ \\
 enumerate $(\lambda_0,\ldots, \lambda_{n-1},1)$
}
}
}
}
\end{algorithm}
\end{proposition}

\begin{proposition}
The Algorithm~\ref{algo:randomSymmetric} provides random generation of coincident symmetric functions.
\begin{algorithm}[h]
  \caption{\label{algo:randomSymmetric}
    UniformRandomGeneration}
  \KwIn{The number of variables: $n \in \N$.}
\BlankLine
\KwOut{Uniform random generation of a coincident symmetric functions with $n$ variables}
\If { $n = 1$}{
  \Return uniform( $\{ (0,0), (0,1)\}$)\\
}
\Else {
  \If {$n$ even}{
    $(\lambda_0,\ldots, \lambda_{n-1}) \gets $  UniformRandomGeneration($n-1$)\\
    $\lambda_n \gets$ uniform $(\{ 0, 1\})$\\
    \Return  $(\lambda_0,\ldots, \lambda_n)$. 
    }
    \Else { 
      $(\lambda_0,\ldots, \lambda_{n-1}) \gets$  UniformRandomGeneration($n-1$) \\
      $\lambda_n \gets$ uniform $(\{ 0, 1 \})$.\\
      \If { $\lambda_0 \; p(n,n) \oplus \lambda_2 \; p(n,n-2) \ldots \oplus \lambda_{n-2} \; p(n,2) = \lambda_{n-1}$}{
        \Return  $(\lambda_0,\ldots, \lambda_{n-1},\lambda_n)$. 
        }
      \Else {
        \Return  $(\lambda_0,\ldots, 1- \lambda_{n-1},\lambda_n)$. 
        }
      }
    }
\end{algorithm}
\end{proposition}

\section{Experiments results}
\label{sec:experiment}

The usefulness of coincident functions for practical applications is
not yet established.  Therefore a deep investigation of other
cryptographically significant properties of Boolean functions must be
conducted from a cryptographic point of view.  We show in this part
that for some aspects a random coincident function looks like uniform
random Boolean functions. We consider the Hamming weight, the distribution of
the degrees (the number of monomials for each degree), the balancedness and the
nonlinearity.  Other investigation like propagation criteria,
algebraic immunity may also be considered.

\subsection{Correlation-immune Boolean functions}
Correlation immunity of a Boolean function is a measure of the degree
to which its outputs are uncorrelated with some subset of its inputs.
In \cite{siegenthaler1984}, the author shows the importance of this property in cryptography.

The table below gives the number of correlation-immune functions of order $1$ for fix Hamming weight 
for $n \le 5$. We write $cor_1(n)$ the total number of
correlation-immune functions of order $1$ with $n$ variables.

Remark that the only $1$-resilient in this table is the function with $2$ variables $x_1 \oplus x_2$.

\[
\begin{array}{c|c|c|c|c|c|c|c|c|c|c|c|c|c}
n \setminus m & 0 & 2 & 4 & 6 & 8 & 10 & 12 & 14 & 18 & 22 & 30 & \mbox{Total}& cor_1(n)\\
\hline
\hline
1 & 1 & & & & & & & & & & & 1 & \\
2 & 1 & 1 & & & & & & & & & & 2 & 4 \\
3 & 1 & & 1 & &  & & & & & & & 2 & 18 \\
4 & 1 & 3 & & & 3 & & 1 &  & & & & 8 & 648 \\
5 & 1 & & & 5 & & 70 & & & 70 & 5 & 1 & 152 & 3140062
\end{array}
\]

\subsection{Hamming weight distribution}
We observe that $\C_n$ and $\F_n$ follows a similar weight distribution.
The Figure~\ref{hammingWeight} gives the distribution of Hamming
weight for $1000$ uniform random generated Boolean functions over
$\C_n$ and $\F_n$, with $n = 20$.  The Hamming weight $w$ is normalized
by the mean $2^{n-1}$, so the abscissa will be $w' = w/2^{n-1}$.
It is easily shown that the distributions are very similar. 
\begin{figure} \label{weight}
  \centering
  \begin{tabular}{cc}
    \scalebox{.3}{\includegraphics{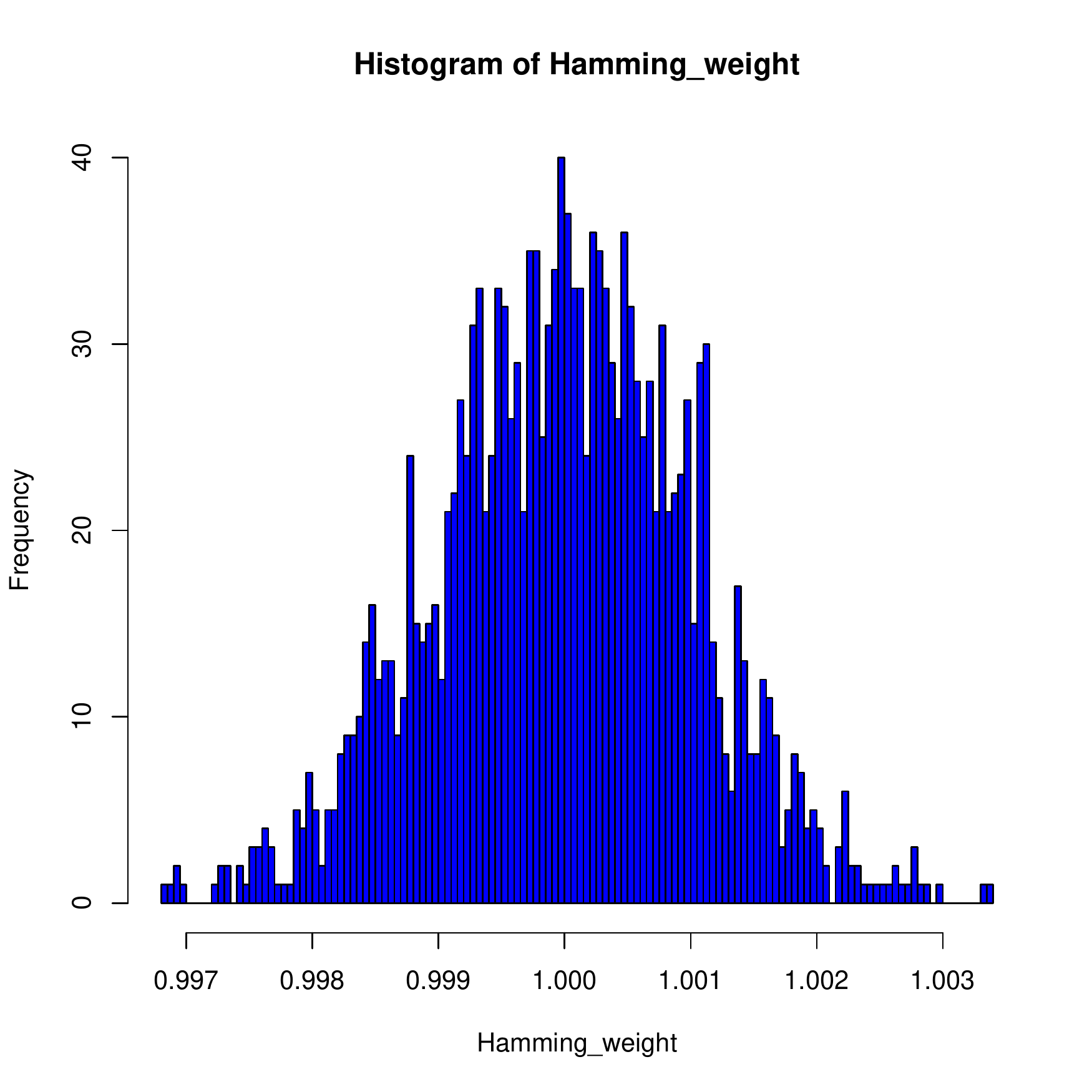}}
    & \scalebox{.3}{\includegraphics{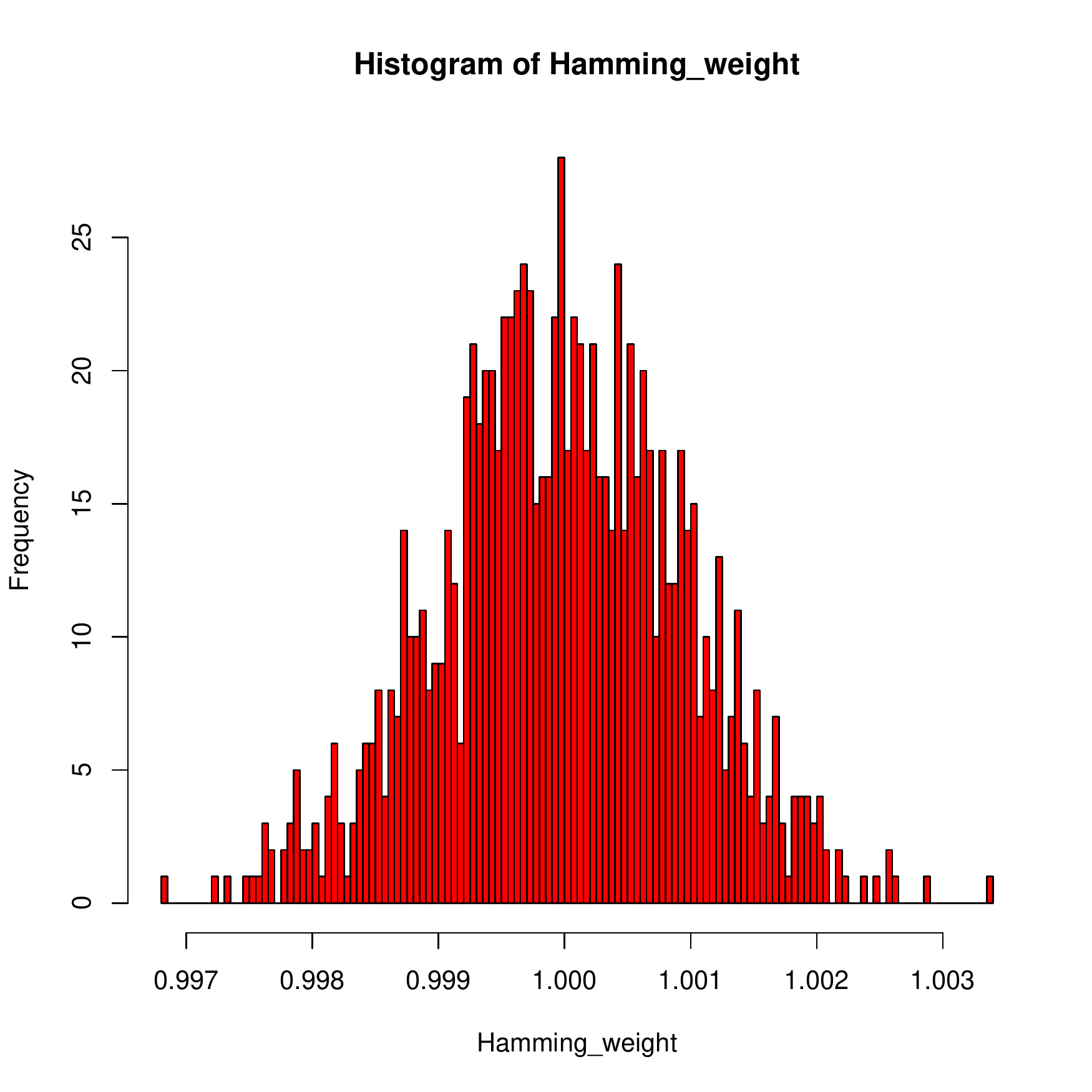}}\\
    Random functions over $\C_{20}$ & Random functions over $\F_{20}$
  \end{tabular}
  \caption{\label{hammingWeight}
  Hamming weight distributions for coincident and boolean functions with 20 variables.}
\end{figure}

\subsection{Algebraic degree distribution}
We show experimentally that random coincident functions follow the same algebraic degree distribution 
than any random Boolean function.
Let $f$ and $g$ be Boolean functions uniformly randomly generated over respectively
$\C_n$ and $\F_n$, we consider the number of monomials of degree $d$
for each $d \in \{0, \ldots, n\}$. For any $d$, the expected value
should be close to ${ n \choose d}/2$ in the case of $g$. We have used
several times the Kolmogorov-Smirnov statistical test to show that the
distributions are very closed.
Of course if we consider just one degree $d$, its is possible to distinguish the distributions 
since in $f$ we have either all the monomials of degree $1$ or neither (Proposition~\ref{prop:coincidentDegree1}).
The Figure~\ref{degree} in a single random generation of $f$ and $g$. 

\begin{figure} \label{degree}
  \centering
  \begin{tabular}{cc}
    \scalebox{.3}{\includegraphics{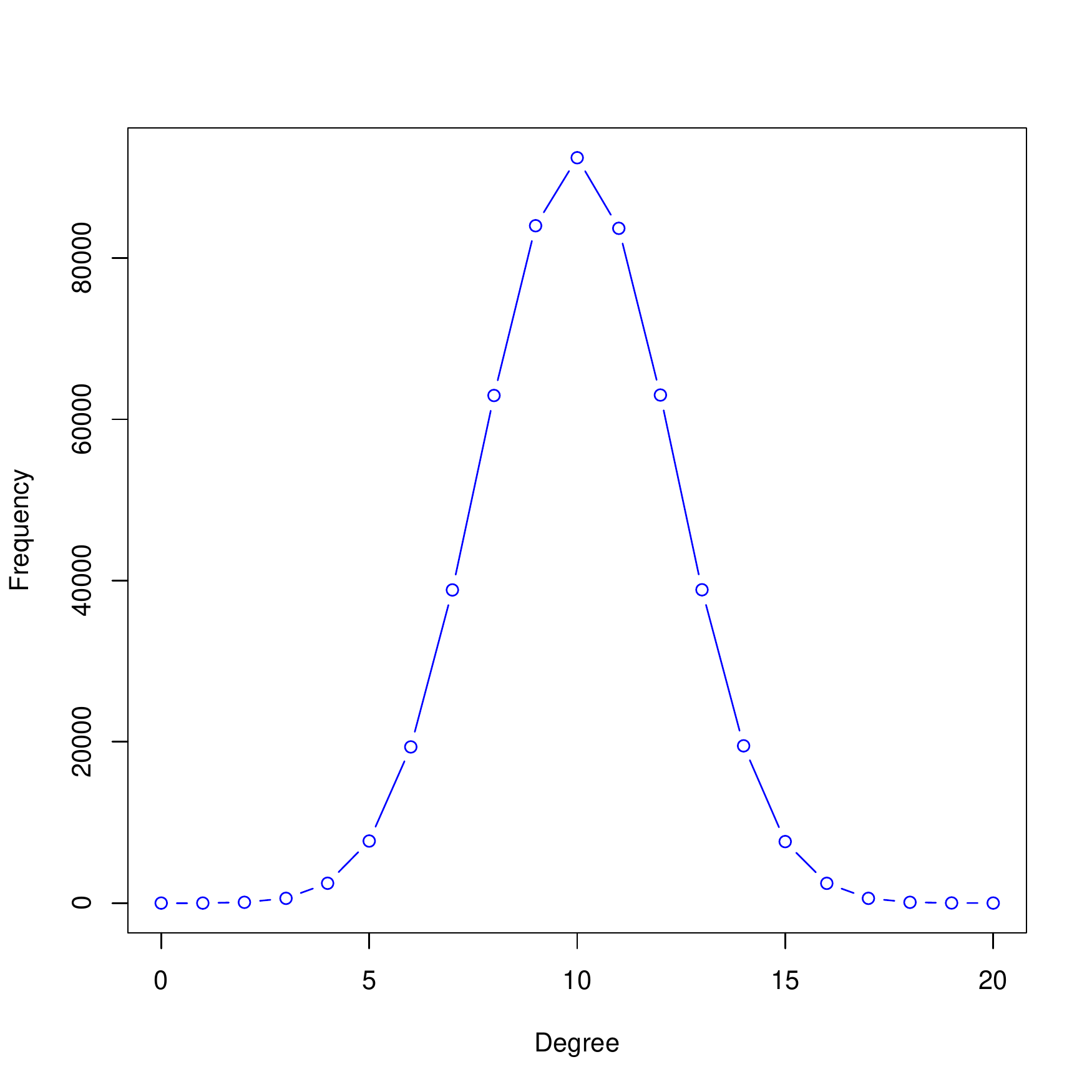}}
    & \scalebox{.3}{\includegraphics{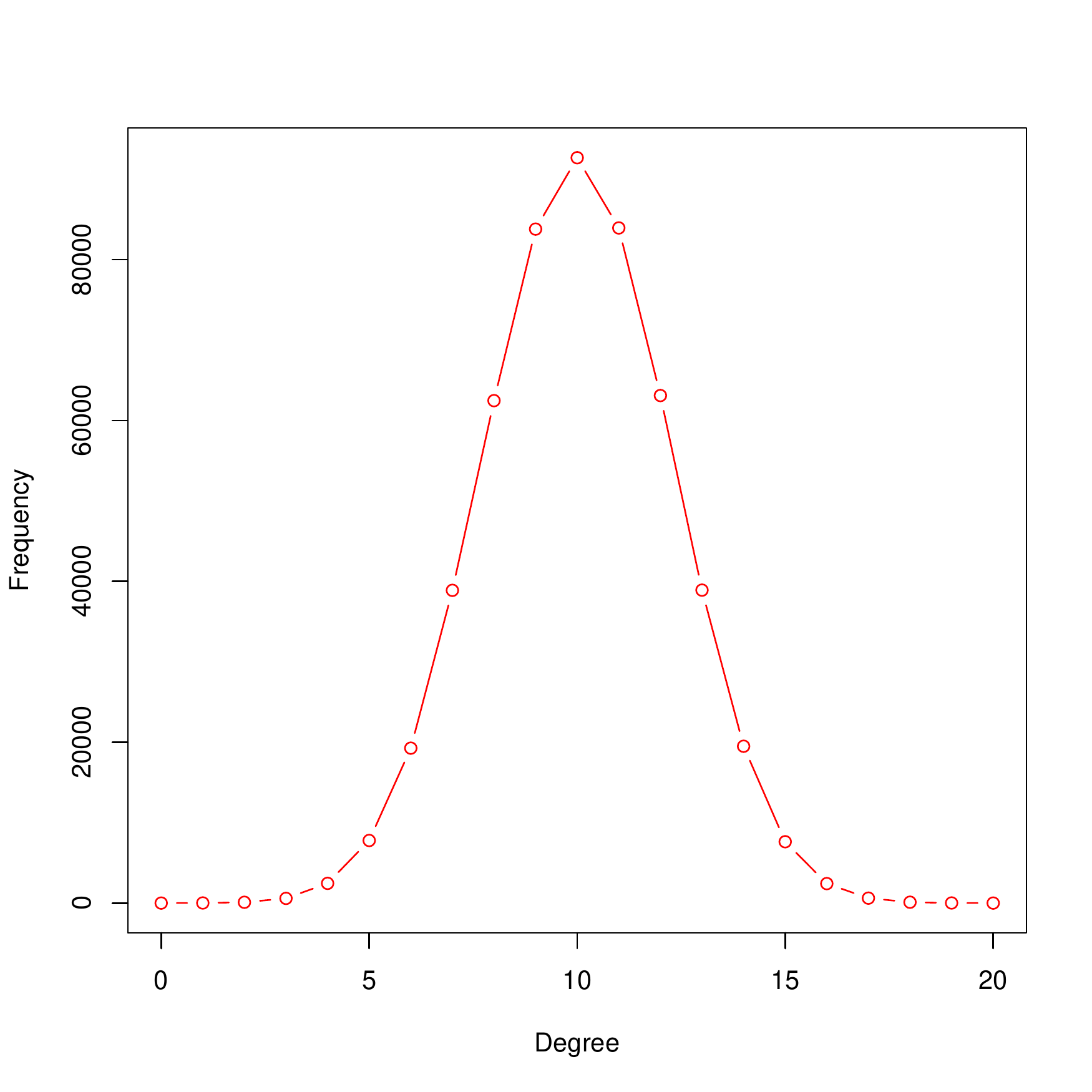}}\\
    Random function over $\C_{20}$ & Random functions over $\F_{20}$
  \end{tabular}
  \caption{\label{hammingWeight}
  Degree distribution of Boolean functions with 20 variables.}
\end{figure}

\subsection{Balancedness}
A Boolean function of $n$ variables is balanced when $w_H(f_n) = 2^{n-1}$. 
It is easily seen that the probability that a random Boolean function is balanced is equal to 
$\frac{ { 2^{n} \choose 2^{n-1}}}{2^{2^n}}$. We obtain every close frequencies when we generates 
random coincident function under $n \le 20$. Further investigations should proved this statement 
for any $n \in \N$.

\subsection{Non-linearity}
The non-linearity is the Hamming weight distance to all the affine
Boolean functions.  It is shown that Boolean functions must have a
high nonlinearity to ensure confusion, an important property in
cryptographic context(\cite{carletBook}), in particular to avoid fast
correlation attacks(\cite{canteaut2000,chepyzhov1991}).  The functions
which reach the best nonlinearity are called bent functions, they
occur when $n$ is even and they have nonlinearity $2^{n-1} - 2^{n/2-1}$. 
On the other hand, the best nonlinearity of Boolean functions in odd numbers of
variables is strictly greater than the quadratic bound --$2^{n-1} - 2^{\frac{n-1}{2}}$-- for any $n > 7$.
See (\cite{carletBook}) for a good introduction of the nonlinearity property and the bent functions.

Let $h \in \C_n$, there exists a unique $g \in \mathcal F_{n-1}$ such
that \mbox{$h = (1 \oplus x_n) \varphi_{n-1}(g) \oplus x_n g$}
(Proposition~\ref{prop:equivCoincident}).  Let $l_n \in \mathcal L_n$,
the set of affine functions with $n$ variables.  There is $l_{n-1} \in
\mathcal L_{n-1}$ such that $l_n = (1\oplus x_n) l_{n-1} \oplus x_n
l_{n-1}$ or \mbox{$l_n = (1\oplus x_n) l_{n-1} \oplus x_n (l_{n-1} \oplus
1)$}.  Hence $d(h,l_n) = w_H(\varphi_{n-1}(g),l_{n-1})+w_H(g,l_{n-1})$
or $d(h,l_n) = w_H(\varphi_{n-1}(g),l_{n-1})+w_H(g\oplus 1,l_{n-1})$.
A random $h$ of $\mathcal C$ should have the same distance from the
affine functions from any random Boolean function if the there is no
correlation between the distances $w_H(\varphi_{n-1}(g),l_{n-1})$ 
and $w_H(g\oplus 1,l_{n-1})$. 

We observe experimentally that this is the case. 
In Figure~\ref{nonlinearity}, we consider $n = 10$ and $11$ and we build $1000$ random functions over $\mathcal C_{n}$ and 
$1000$ over $\mathcal F_{n}$ and we compute the frequency of each nonlinearity value.
Remark that bent functions have nonlinearity $2^{9}-2^4 = 496$. 
We consider even and odd values of $n$ because the nonlinearity behaviour is very different in these cases.

\begin{figure}
  \centering
  \begin{tabular}{cc}
    \scalebox{.2}{\includegraphics{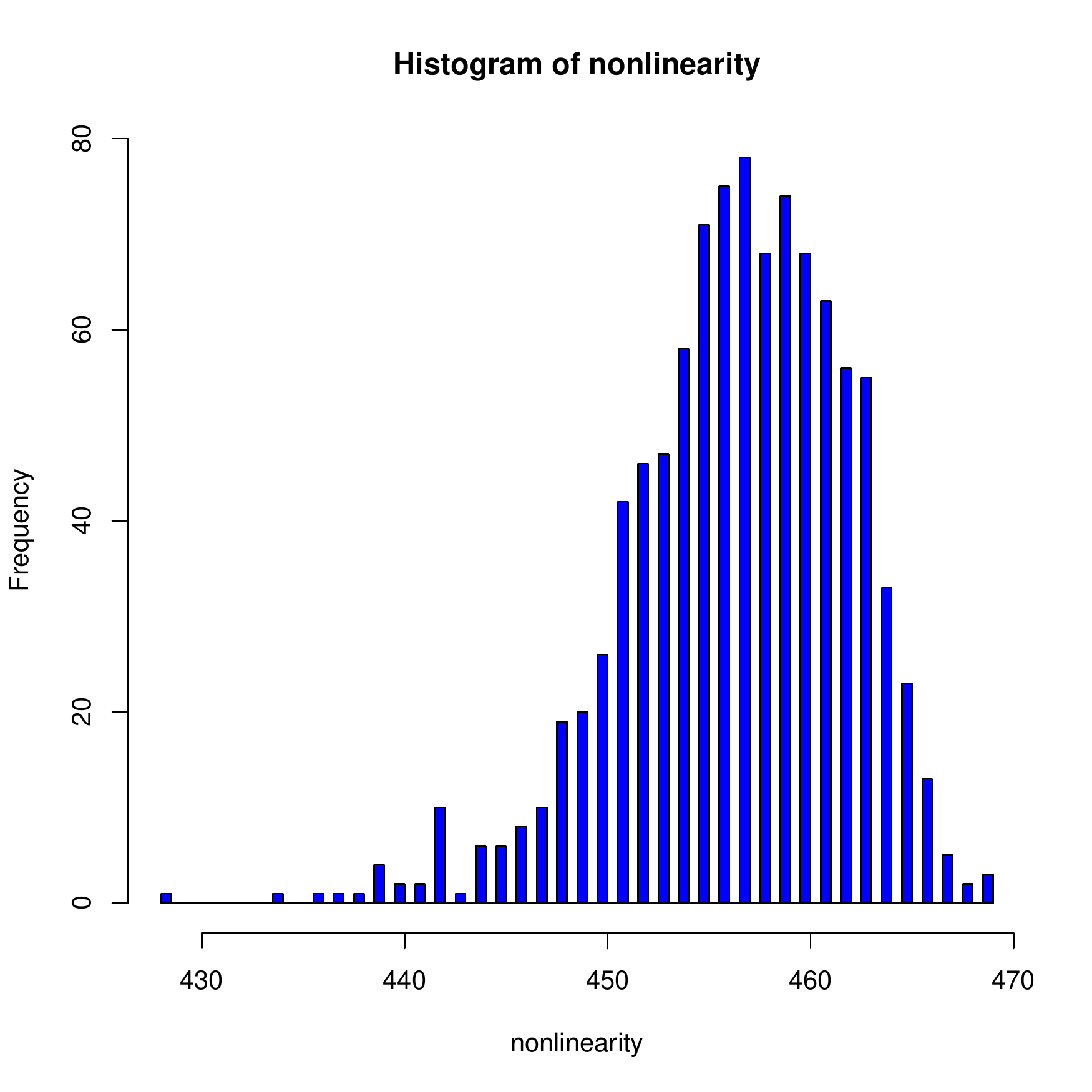}}
    & \scalebox{.2}{\includegraphics{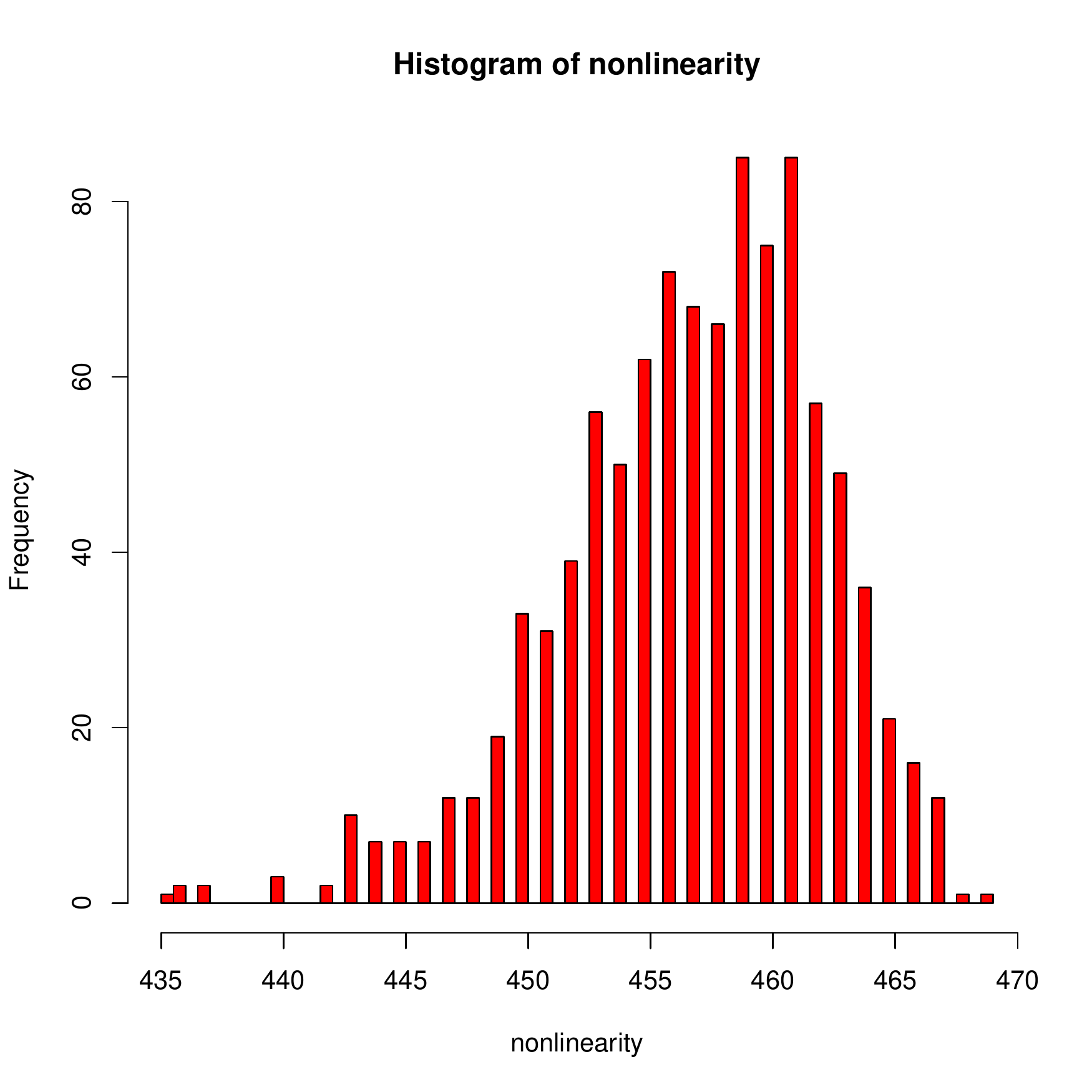}}\\
    Random functions over $\C_{10}$ & Random functions over $\F_{10}$\\
  \scalebox{.2}{\includegraphics{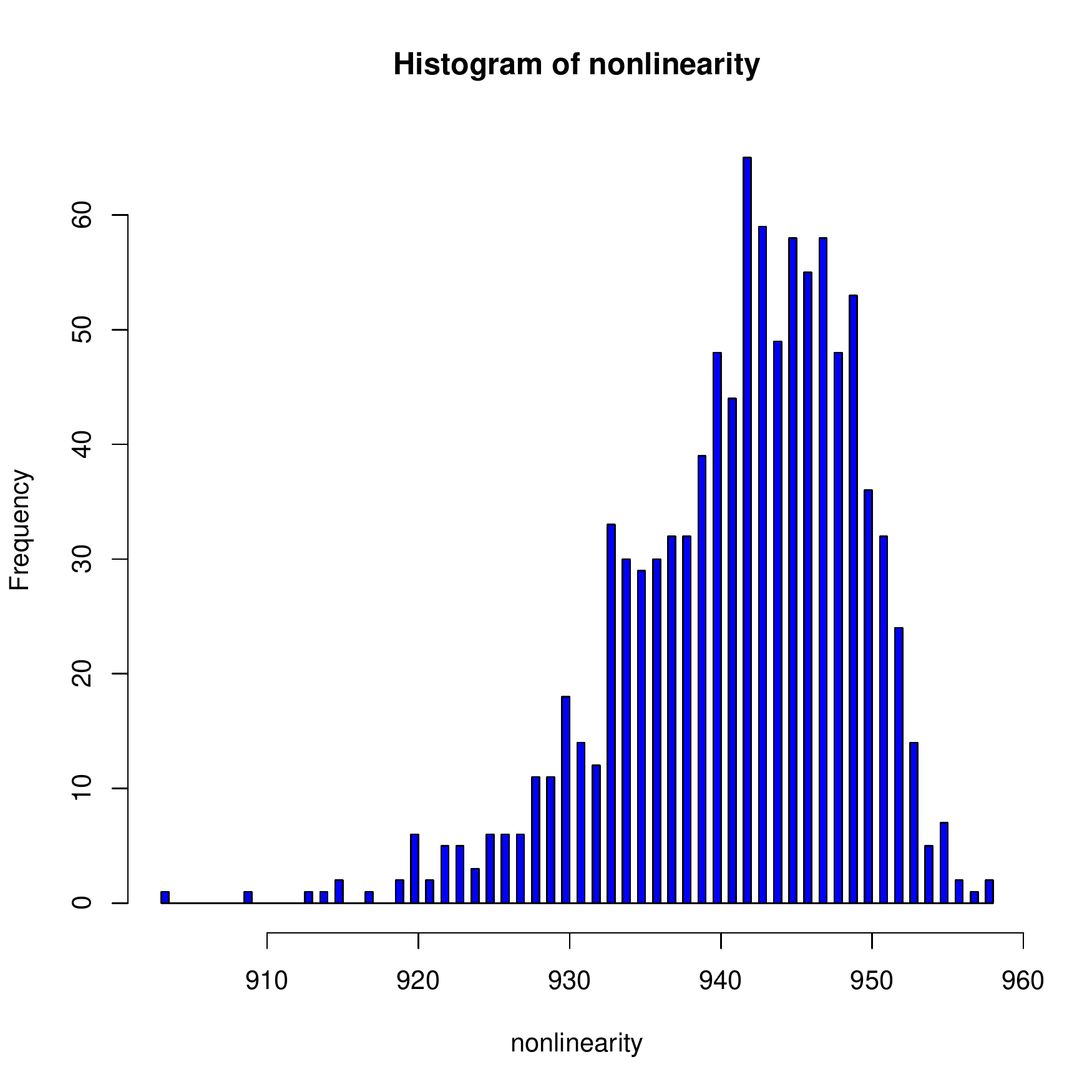}}
    & \scalebox{.2}{\includegraphics{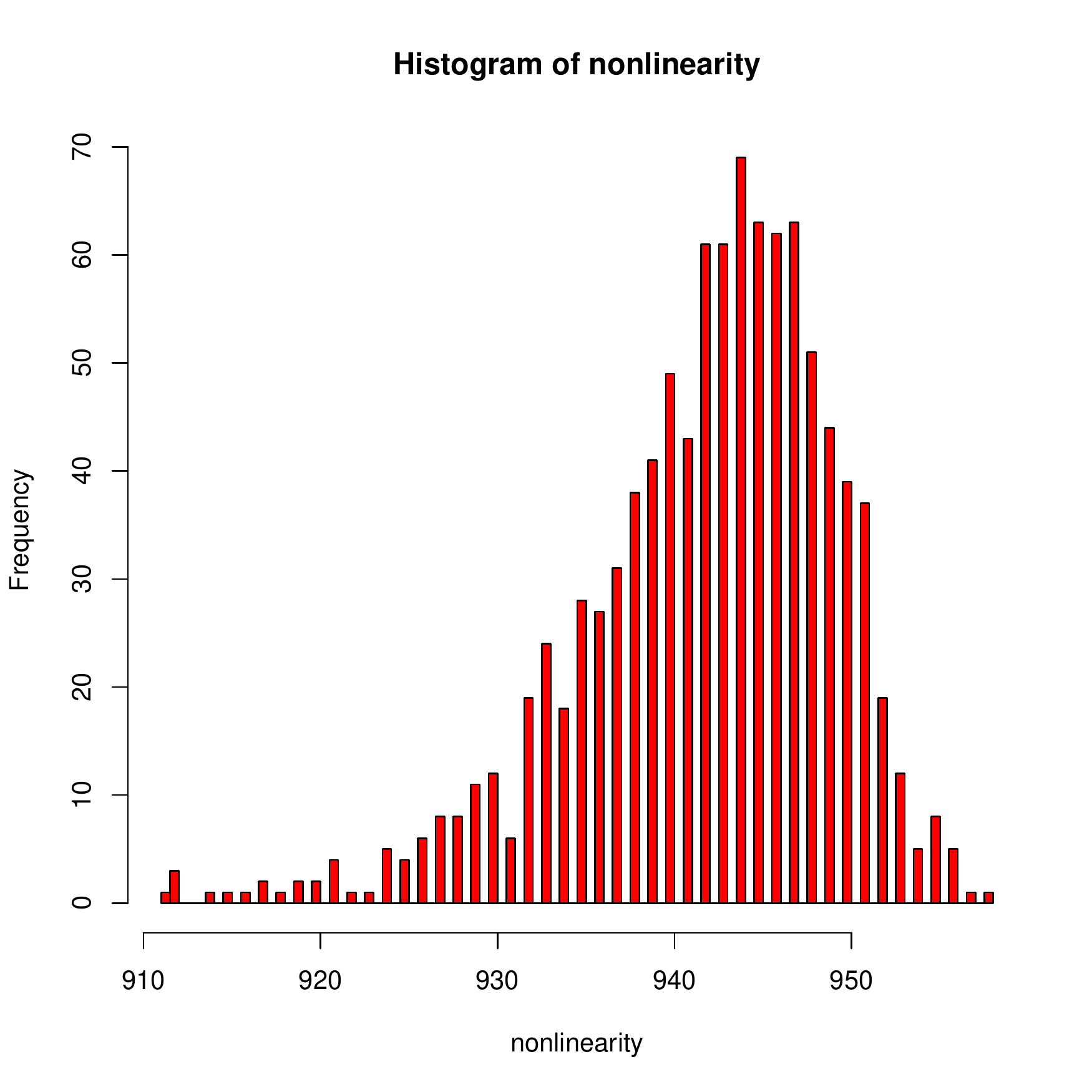}}\\
    Random functions over $\C_{11}$ & Random functions over $\F_{11}$
  \end{tabular}
  \caption{\label{nonlinearity}
    Nonlinearity distributions for coincident and boolean function with $10$ and $11$ variables.}
\end{figure}

\section{Conclusion}
\label{sec:conclusion}

Our paper presents an innovative way to manipulate the Mobius
transform, with the distinction between variables abd indeterminates.
This allows us to highlight new properties of coincident functions.
We may also move easily from Shannon decomposition to Reed-Muller
decomposition or vice versa. We show how the Boolean lattice is
colored in the case of coincident functions, which gives a method to
build monotone functions, and we provide a method to build and
generate uniformly the symmetric coincident functions.  Thanks to a
uniform random generator over all the coincident boolean functions, we
establish experimentally that for the most common characteristics
(Hamming weight, distributions of the monomials degree, balanceness),
a coincident function looks like any Boolean function. This
experimental work could be completed by further properties. Notice
that our random generator of coincident functions requires the
computation of a Mobius transform with $n-1$ variables.  We may avoid
this problem with the use of the basis of coincident functions but it
is not an efficient way. Direct random generation and enumeration will
be a challenge.  Another promising perspective will be to propose
algorithms which compute the Mobius transform with low complexity for
a larger part of Boolean functions, standed by our new
properties. Based on all these results, we therefore recommand to use
this class of functions, especially in order to build Boolean
functions with good cryptographic properties. Specific constructions
with trade-off between cryptographic criteria seems really feasible.

\bibliographystyle{spr-chicago}
\bibliography{biblio}

\end{document}